\documentclass{article}%
\usepackage{amsmath}
\usepackage{amsfonts}
\usepackage{amssymb}
\usepackage{graphicx}%
\providecommand{\U}[1]{\protect\rule{.1in}{.1in}}
\newtheorem{theorem}{Theorem}[section]

\newtheorem{example}[theorem]{Example}

\newtheorem{lemma}[theorem]{Lemma}

\newtheorem{proposition}[theorem]{Proposition}

\newenvironment{proof}[1][Proof]{\noindent\textbf{#1.} }{\ \rule{0.5em}{0.5em}}
\begin{document}

\title{When is an input state always better than the others ?: universally optimal
input states for statistical inference of quantum channels}
\author{Keiji Matsumoto\\National Institute of Informatics, HItotsubashi 2-1-2, Chiyoda-ku, \\Tokyo 101-8430, e-mail: keiji@nii.ac.jp}
\maketitle

\begin{abstract}
Statistical estimation and test of unknown channels have attracted interest of
many researchers. In optimizing the process of inference, an important step is
optimization of the input state, which in general do depend on the kind of
inference (estimation or test, etc.), on the error measure, and so on. But
sometimes, there is a universally optimal input state, or an input state best
for all the statistical inferences and for all the risk functions. In the
paper, the existence of a universally optimal state is shown for group
covariant/contravariant channels, unital qubit channels and some measurement
families. To prove these results, theory of "comparison of state families" are
used. We also discuss about effectiveness of entanglement and adaptation of
input states.

\end{abstract}

\section{Introduction}

Statistical estimation and test of unknown channels have attracted interests
of many researchers. Below, let $\left\{  \Lambda_{\theta}\right\}
_{\theta\in\Theta}$ be a family of unknown channels, where $\theta\in\Theta$
is the unknown parameter. In optimizing the process of inference, one has to
optimize not only the measurement performed upon the output state
$\Lambda_{\theta}\otimes\mathbf{I}\left(  \rho_{in}\right)  $, but also the
input state $\rho_{in}$. (One may also use a process POVM\thinspace
\cite{Ziman}, operators $\left\{  M_{t}\right\}  _{t\in\mathcal{D}}$ such that
$\sum_{t\in\mathcal{D}}M_{t}=\mathbf{1}\otimes\mathrm{tr}\,_{\mathcal{H}_{R}%
}\rho_{in}^{T}$ . But then one also has to optimize $\mathrm{tr}%
\,_{\mathcal{H}_{R}}\rho_{in}$. Since an optimal input state $\rho_{in}$ is a
pure state, optimization of $\mathrm{tr}\,_{\mathcal{H}_{R}}\rho_{in}$ is
equivalent to optimization of $\rho_{in}$ .)

In general, optimal input states depend on whether we are estimating state or
testing hypothesis about unknown channels; they also depend on error measure,
and detail of the setting (Bayesian, minimax, unbiased estimation,
Neyman-Pearson test, etc.).

In some cases, however, the situation is less complicated. For example,
\cite{ChiribellaDarianoSacchi} deals with estimation of group transform
$\left\{  U_{g}\right\}  _{g\in\mathcal{G}}$, where $g\rightarrow U_{g}$ is a
representation of the group $\mathcal{G}$ and $g$ is unknown and to be
estimated. They had shown that there is an input state which is optimal with
respect to any $\mathcal{G}$-invariant loss functions. (In case of
$\mathcal{G}=\mathrm{SU}\left(  d\right)  $ and $U_{g}=g$, maximally entangled
states between the input space and the auxiliary space are optimal.) Meantime,
\cite{Fujiwara:2002} treats estimation of $\mathrm{SU}\left(  2\right)  $
channel by an unbiased estimator, and `the loss function' here is the mean
square error matrix of the estimate $\hat{\theta}$ of the unknown real vector
$\theta$ which parameterizes $\mathcal{G}=\mathrm{SU}\left(  2\right)  $.
Since the space of matrices is not totally ordered, the existence of the
minimum is non-trivial. Put differently, if the loss is scalar valued
increasing function of a mean square error matrix, then, maximally entangled
states are optimal. Also, \cite{Sacchi:2005} studies discrimination of a pair
of generalized Pauli matrices, and shows maximally entangled states minimize
Bayesian error probability for any prior distributions. In case of qubits,
they extended their result to minimax error probability\thinspace
\cite{Sacchi:2005:3}. Another example of such study is \cite{ZimanSedlak},
where discrimination of two unitary operation is discussed. They found that
minimizers of Bayesian error probability and the error probability of
unambiguous discrimination are the same.

These results motivate the following definition: we say the input is
\textit{universally optimal }for the family $\left\{  \Lambda_{\theta
}\right\}  _{\theta\in\Theta}$, roughly speaking, if it is optimal for all the
statistical inferences and for all the loss functions. (The rigorous
definition will be given later.) We show that a universally optimal state
exists (not necessarily uniquely) \ in case of group covariant and
contravariant channels, unital qubit channels and some measurement families.

To prove these results, we have recourse to the theory of "comparison of state
families" \cite{Buscemi:12}\cite{Matsumoto}; we write $\left\{  \rho_{\theta
}\right\}  _{\theta\in\Theta}\succeq^{c}\left\{  \sigma_{\theta}\right\}
_{\theta\in\Theta}$ if the family $\left\{  \rho_{\theta}\right\}  _{\theta
\in\Theta}$ is more informative than another family $\left\{  \sigma_{\theta
}\right\}  _{\theta\in\Theta}$ with respect to any kind of statistical
inferences. Then, our target is to prove
\[
\forall\rho^{\prime}\,\,\left\{  \left(  \Lambda_{\theta}\otimes
\mathbf{I}\right)  \left(  \rho_{\mathrm{opt}}\right)  \right\}  _{\theta
\in\Theta}\succeq_{c}\left\{  \left(  \Lambda_{\theta}\otimes\mathbf{I}%
\right)  \left(  \rho^{\prime}\right)  \right\}  _{\theta\in\Theta}%
,\,\forall\rho^{\prime}%
\]
for an input $\rho_{\mathrm{opt}}$. In particular, we utilize sufficient
conditions for $\left\{  \rho_{\theta}\right\}  \succeq^{c}\left\{
\sigma_{\theta}\right\}  $, Proposition\thinspace\ref{prop:randomize} and
Lemma\thinspace\ref{lem:commute}.

Based on these results, some related topics are discussed. The first topic is
effect of entanglement between the input space and the auxiliary space. For
example, in \cite{Sacchi:2005}\cite{Sacchi:2005:2}\cite{Sacchi:2005:3}, they
study the condition that Bayes risk and minimax risk of discrimination of two
unital qubit channels is smaller on an entangled state than on any separable
state. In our case, in Sections\thinspace\ref{sec:covariant}%
-\ref{sec:measurement} it is shown that a maximally entangle is universally
optimal for some channel families. But there might be a separable state which
is as good as maximally entangled states. So we question whether the
entanglement is really needed or not.

The second topic discussed is the existence of universally optimal states
under the setting where the given channel can be used for several times.

The paper is organized as follows. In Section\thinspace\ref{sec:preliminaries}%
, besides introducing notations and definitions, the theory of comparison of
state families is explained. In Sections \ref{sec:2-unitary}, \thinspace
\ref{sec:covariant} and \ref{sec:d=2}, universally optimal input states are
established for a pair of unitary operations, covariant/contravariant channel
families, and unital qubit channel families, respectively. In the proof,
Proposition\thinspace\ref{prop:randomize} is used. In Section
\ref{sec:measurement}, with the help of Lemma\thinspace\ref{lem:commute}, we
investigate universally optimal states for some families which consist of a
pair of measurements. In Section\thinspace\ref{sec:su(d)}, the family of
$\mathrm{SU}\left(  d\right)  $ is studied. In $d=2$-case, it is shown, with
recourse to Theorem\thinspace\ref{th:unital-qubit}\ in Section\thinspace
\ref{sec:d=2}, that maximally entangled states are universally optimal. On the
other hand, in $d\geq3$-case, it is shown that any state is optimal for some
statistical inferences. In Section\thinspace\ref{sec:entanglement}, we
investigate the conditions that an entangled state is strictly universally
better than any separable states. In Section\thinspace\ref{sec:iteration},
universally optimal input states in case of iterative use of the given channel
is studied.

\section{Preliminaries}

\label{sec:preliminaries}

\subsection{Settings, conventions and notations}

Below, $\mathcal{H}_{in}$, $\mathcal{H}_{out}$, $\mathcal{H}_{R}$ etc. are
finite dimensional Hilbert spaces, and $\mathcal{B}\left(  \mathcal{H}%
_{in}\right)  $, for example, are the set of linear operators over
$\mathcal{H}_{in}$. $\mathbf{1}_{in}$ and $\mathbf{I}_{in}$ denotes identity
operator over $\mathcal{H}_{in}$ and over $\mathcal{B}\left(  \mathcal{H}%
_{in}\right)  $, respectively. A channel is a trace preserving completely
positive (CPTP, hereafter) map from $\mathcal{B}\left(  \mathcal{H}%
_{in}\right)  $ to $\mathcal{B}\left(  \mathcal{H}_{out}\right)  $, and is
represented by $\Lambda$, $\Upsilon$ , etc. with subscripts and superscripts.

To do some statistical inference about a family $\left\{  \Lambda_{\theta
}\right\}  _{\theta\in\Theta}$ of channels $\Lambda_{\theta}:\mathcal{B}%
\left(  \mathcal{H}_{in}\right)  \rightarrow\mathcal{B}\left(  \mathcal{H}%
_{out}\right)  $, a statistician prepares an input state $\rho_{in}%
\in\mathcal{B}\left(  \mathcal{H}_{in}\otimes\mathcal{H}_{R}\right)  $, sends
its $\mathcal{H}_{in}$-part to $\Lambda_{\theta}$, obtaining $\Lambda_{\theta
}\otimes\mathbf{I}_{R}\left(  \rho_{in}\right)  $ as the output. To the output
$\Lambda_{\theta}\otimes\mathbf{I}_{R}\left(  \rho_{in}\right)  $, the
statistician apply a measurement with POVM $M$ which takes values in decision
space $\mathcal{D}$ (an element of $\mathcal{D}$ is usually denoted by $t$).
Without loss of generality, throughout the paper, we suppose $\rho_{in}$ is
pure, and thus we suppose $\dim\mathcal{H}_{in}=\dim\mathcal{H}_{R}=d$. For a
state vector $\left\vert \psi\right\rangle \in$\ $\mathcal{H}_{in}%
\otimes\mathcal{H}_{R}$,
\begin{equation}
\rho_{\psi}:=\mathrm{tr}_{\mathcal{H}_{R}}\,\left\vert \psi\right\rangle
\left\langle \psi\right\vert \in\mathcal{B}\left(  \mathcal{H}_{in}\right)  .
\label{rho-psi}%
\end{equation}
A system of vectors $\left\{  \left\vert i\right\rangle \right\}  _{i=1}^{d}$
is an orthonormal complete basis of $\mathcal{H}_{in}$ . Abusing the notation,
the same symbol is also used to denote an orhonormal basis of $\mathcal{H}%
_{R}$.
\[
\left\vert \Phi_{d}\right\rangle :=\frac{1}{\sqrt{d}}\sum_{i=1}^{d}\left\vert
i\right\rangle \left\vert i\right\rangle
\]
is a maximally entangled state living in $\mathcal{H}_{in}\otimes
\mathcal{H}_{R}$.

Given a linear map $\Gamma$ from $\mathcal{B}\left(  \mathcal{H}_{in}\right)
$ to $\mathcal{B}\left(  \mathcal{H}_{out}\right)  $, its Choi-Jamilokovski's
representation $Ch\left(  \Gamma\right)  $ is defined by
\[
Ch\left(  \Gamma\right)  :=\sum_{i,j=1}^{d}\Gamma\left(  \left\vert
i\right\rangle \left\langle j\right\vert \right)  \otimes\left\vert
i\right\rangle \left\langle j\right\vert .
\]
We also use the following notation:
\[
\Upsilon_{C}\left(  \rho\right)  :=C\rho C^{\dagger}.
\]
Given a state $\rho$ and a POVM $M$, denote $P_{\rho}^{M}\left(  B\right)
:=\mathrm{tr}\,\rho M\left(  B\right)  $.

When $\Theta\subset\mathcal{D}=%
\mathbb{R}
^{m}$, we write
\begin{align*}
\mathrm{E}\left[  M,\rho_{\theta}\right]   &  :=\int_{t\in\mathcal{D}%
}t\,\mathrm{d}P_{\rho_{\theta}}^{M}\left(  t\right)  ,\\
\mathrm{V}\left[  M,\rho_{\theta}\right]   &  :=\left[  \int\left(
t^{i}-\theta^{i}\right)  \left(  t^{j}-\theta^{j}\right)  \mathrm{d}%
P_{\rho_{\theta}}^{M}\left(  t\right)  \right]  .
\end{align*}

\subsection{Comparison of state families}

In comparison of input states, we have recourse to the theory of comparison of
state families\cite{Buscemi:12}\cite{Matsumoto}. Consider a family $\left\{
\rho_{\theta}\right\}  _{\theta\in\Theta}$ of states over $\mathcal{H}$ and a
family $\left\{  \sigma_{\theta}\right\}  _{\theta\in\Theta}$ of states over
$\mathcal{H}^{\prime}$. We say $\left\{  \rho_{\theta}\right\}  _{\theta
\in\Theta}$ is \textit{sufficient} to $\left\{  \sigma_{\theta}\right\}
_{\theta\in\Theta}$ with respect to classical decision problems, and write
$\left\{  \rho_{\theta}\right\}  _{\theta\in\Theta}\succeq^{c}\left\{
\sigma_{\theta}\right\}  _{\theta\in\Theta}$, if and only if, for any decision
space $\mathcal{D}$ equipped with $\sigma$-field $\mathfrak{A}$, any $\sigma
$-field $\mathfrak{B}$ over $\Theta$, any loss function $l:\Theta
\times\mathcal{D}\rightarrow%
\mathbb{R}
_{+}$ which is jointly measurable, any probability measure $\pi$ over $\left(
\Theta,\mathfrak{B}\right)  $, and for any measurement $M^{\prime}$ over
$\left(  \mathcal{D},\mathfrak{A}\right)  $ in $\mathcal{H}^{\prime}$, there
is a measurement $M$ over $\left(  \mathcal{D},\mathfrak{A}\right)  $ in
$\mathcal{H}$ such that
\[
\int_{\Theta\times\mathcal{D}}l_{\theta}\left(  t\right)  \mathrm{d}%
P_{\rho_{\theta}}^{M}\left(  t\right)  \mathrm{d}\pi\left(  \theta\right)
\leq\int_{\Theta\times\mathcal{D}}l_{\theta}\left(  t\right)  \,\mathrm{d}%
P_{\sigma_{\theta}}^{M^{\prime}}\left(  t\right)  \mathrm{d}\pi\left(
\theta\right)  .
\]
When $\left\{  \rho_{\theta}\right\}  _{\theta\in\Theta}\succeq^{c}\left\{
\sigma_{\theta}\right\}  _{\theta\in\Theta}$ and $\left\{  \sigma_{\theta
}\right\}  _{\theta\in\Theta}\succeq^{c}\left\{  \rho_{\theta}\right\}
_{\theta\in\Theta}$ holds, we write $\left\{  \rho_{\theta}\right\}
_{\theta\in\Theta}\equiv^{c}\left\{  \sigma_{\theta}\right\}  _{\theta
\in\Theta}$.

\begin{lemma}
\label{lem:sufficient}\bigskip$\left\{  \rho_{\theta}\right\}  _{\theta
\in\Theta}\succeq^{c}\left\{  \sigma_{\theta}\right\}  _{\theta\in\Theta}$
holds if and only if, for any measurement $M$ on $\left(  \mathcal{D}%
,\mathfrak{A}\right)  $, there is a measurement $M^{\prime}$ on $\left(
\mathcal{D},\mathfrak{A}\right)  $ such that $P_{\rho_{\theta}}^{M^{\prime}%
}=P_{\sigma_{\theta}}^{M}$.
\end{lemma}

Due to Lemma\thinspace\ref{lem:sufficient}, $\left\{  \rho_{\theta}\right\}
_{\theta\in\Theta}\succeq^{c}\left\{  \sigma_{\theta}\right\}  _{\theta
\in\Theta}$ has very strong implications: whatever the settings are, and
whatever the error measures are chosen, $\left\{  \rho_{\theta}\right\}
_{\theta\in\Theta}$ is always better than $\left\{  \sigma_{\theta}\right\}
_{\theta\in\Theta}$. For example, for any decision space $\mathcal{D}$
equipped with $\sigma$-field $\mathfrak{A}$, any loss function $l:\Theta
\times\mathcal{D}\rightarrow%
\mathbb{R}
_{+}$ such that $l_{\theta}\left(  \cdot\right)  $ is measurable, the minimax
risk is always smaller on $\left\{  \rho_{\theta}\right\}  _{\theta\in\Theta}$
than on $\left\{  \sigma_{\theta}\right\}  _{\theta\in\Theta}$ :
\[
\inf_{M}\sup_{\theta\in\Theta}\int_{\mathcal{D}}l_{\theta}\left(  t\right)
\mathrm{d}P_{\rho_{\theta}}^{M}\left(  t\right)  \leq\inf_{M}\sup_{\theta
\in\Theta}\int_{\mathcal{D}}l_{\theta}\left(  t\right)  \mathrm{d}%
P_{\sigma_{\theta}}^{M}\left(  t\right)  .
\]
Also, in hypothesis testing of Neyman-Pearson type, the second error
probability of the optimal level $\alpha$ test is also smaller on $\left\{
\rho_{\theta}\right\}  _{\theta\in\Theta}$ than on $\left\{  \sigma_{\theta
}\right\}  _{\theta\in\Theta}$. That is, letting $\mathcal{D}:=\left\{
0,1\right\}  $,\thinspace\ $\Theta_{0}\cup$ $\Theta_{1}=\Theta$, and
\[
l_{\theta}^{\mathrm{T}}\left(  t\right)  :=\left\{
\begin{array}
[c]{cc}%
1, & \left(  \theta\in\Theta_{0}\text{ and }t=1,\text{ or }\theta\in\Theta
_{1}\text{ and }t=0\right) \\
0, & \text{otherwise}%
\end{array}
\right.  ,
\]
we have \
\begin{align*}
&  \inf_{M}\left\{  \int l_{1}^{\mathrm{T}}\left(  t\right)  \mathrm{d}%
P_{\rho_{\theta}}^{M}\left(  t\right)  ;\int l_{0}^{NP}\left(  t\right)
\mathrm{d}P_{\rho_{\theta}}^{M}\left(  t\right)  \leq\alpha\right\} \\
&  \leq\inf_{M}\left\{  \int l_{1}^{\mathrm{T}}\left(  t\right)
\mathrm{d}P_{\sigma_{\theta}}^{M}\left(  t\right)  ;\int l_{0}^{NP}\left(
t\right)  \mathrm{d}P_{\sigma_{\theta}}^{M}\left(  t\right)  \leq
\alpha\right\}  .
\end{align*}

Another example would be unambiguous discrimination: letting $\mathcal{D}%
:=\left\{  0,1,2\right\}  $,\thinspace\ $\Theta_{0}\cup$ $\Theta_{1}=\Theta$,
\[
l_{\theta}^{\mathrm{IT}}\left(  t\right)  :=\left\{
\begin{array}
[c]{cc}%
\infty, & \left(  \theta\in\Theta_{0}\text{ and }t=1,\text{ or }\theta
\in\Theta_{1}\text{ and }t=0\right) \\
1, & \left(  \theta\in\Theta_{0}\text{ and }t=2,\text{ or }\theta\in\Theta
_{1}\text{ and }t=2\right) \\
0, & \left(  \theta\in\Theta_{0}\text{ and }t=0,\text{ or }\theta\in\Theta
_{1}\text{ and }t=1\right)
\end{array}
\right.  ,
\]
we have
\[
\inf_{M}\int_{\Theta\times\mathcal{D}}l_{\theta}^{\mathrm{IT}}\left(
t\right)  \mathrm{d}P_{\rho_{\theta}}^{M}\left(  t\right)  \mathrm{d}%
\pi\left(  \mathrm{\,}\theta\right)  \leq\inf_{M}\int_{\Theta\times
\mathcal{D}}l_{\theta}^{\mathrm{IT}}\left(  t\right)  \mathrm{d}%
P_{\sigma_{\theta}}^{M}\left(  t\right)  \mathrm{d}\pi\left(  \mathrm{\,}%
\theta\right)  .
\]

Lastly, let $\Theta\subset\mathcal{D}=%
\mathbb{R}
^{m}$. Then, mean square error of an unbiased estimator is always better on
$\left\{  \rho_{\theta}\right\}  _{\theta\in\Theta}$ than on $\left\{
\sigma_{\theta}\right\}  _{\theta\in\Theta}$. That is, for any measurement
$M^{\prime}$ with
\[
\mathrm{E}\left[  M^{\prime},\sigma_{\theta}\right]  =\theta,
\]
there is a measurement $M$ such that
\begin{align*}
\mathrm{E}\left[  M,\rho_{\theta}\right]   &  =\theta,\\
\mathrm{V}\left[  M,\rho_{\theta}\right]   &  =\mathrm{V}\left[  M^{\prime
},\sigma_{\theta}\right]  .
\end{align*}

\begin{proposition}
\label{prop:randomize}\cite{Matsumoto}If there is a trace preserving positive
map $\Gamma$ such that $\Gamma\left(  \rho_{\theta}\right)  =\sigma_{\theta}$,
we have $\left\{  \rho_{\theta}\right\}  _{\theta\in\Theta}\succeq^{c}\left\{
\sigma_{\theta}\right\}  _{\theta\in\Theta}$.
\end{proposition}

\begin{lemma}
\label{lem:alberti-uhlmann}\cite{AlbertiUhlmann}There is a CPTP map $\Gamma$
with
\[
\Gamma\left(  \left\vert \psi_{+}\right\rangle \left\langle \psi
_{+}\right\vert \right)  =\left\vert \varphi_{+}\right\rangle \left\langle
\varphi_{+}\right\vert ,\,\Gamma\left(  \left\vert \psi_{-}\right\rangle
\left\langle \psi_{-}\right\vert \right)  =\left\vert \varphi_{-}\right\rangle
\left\langle \varphi_{-}\right\vert ,
\]
If and only if
\[
\left\vert \left\langle \psi_{+}\right.  \left\vert \psi_{-}\right\rangle
\right\vert \leq\left\vert \left\langle \varphi_{+}\right.  \left\vert
\varphi_{-}\right\rangle \right\vert .
\]

\end{lemma}

\begin{lemma}
\label{lem:commute}\cite{Matsumoto}Suppose $\Theta=\left\{  +,-\right\}  $.
\ If $\left\{  \rho_{\theta}\right\}  _{\theta\in\Theta}\succeq^{c}\left\{
\sigma_{\theta}\right\}  _{\theta\in\Theta}$, then%
\begin{equation}
\left\Vert \rho_{+}-s\,\rho_{-}\right\Vert _{1}\geq\left\Vert \sigma
_{+}-s\,\sigma_{-}\right\Vert _{1},\,\forall s\geq0.\, \label{rho-s rho}%
\end{equation}
If (\ref{rho-s rho}) and $\left[  \rho_{+},\rho_{-}\right]  =0$, then
$\left\{  \rho_{\theta}\right\}  _{\theta\in\Theta}\succeq^{c}\left\{
\sigma_{\theta}\right\}  _{\theta\in\Theta}$.
\end{lemma}

\begin{lemma}
\label{lem:comm-sufficient}Suppose $\Theta=\left\{  +,-\right\}  $,
$\ \sigma_{\theta}\in\mathcal{B}\left(
\mathbb{C}
^{2}\right)  $, and $\left[  \rho_{+},\rho_{-}\right]  =0$. If \ $\left\{
\rho_{\theta}\right\}  _{\theta\in\Theta}\equiv^{c}\left\{  \sigma_{\theta
}\right\}  _{\theta\in\Theta}$, we have
\[
\left[  \sigma_{+},\sigma_{-}\right]  =0.
\]

\end{lemma}

\begin{proof}
By definition, $\left\{  \rho_{\theta}\right\}  _{\theta\in\Theta}\preceq
^{c}\left\{  \sigma_{\theta}\right\}  _{\theta\in\Theta}$ only if there is a
measurement $M$ with
\[
\left\{  \rho_{\theta}\right\}  _{\theta\in\Theta}\preceq^{c}\left\{
P_{\sigma_{\theta}}^{M}\right\}  _{\theta\in\Theta}.
\]
By Lemma\thinspace\ref{lem:commute}, this is equivalent to
\[
\left\Vert \rho_{+}-s\,\rho_{-}\right\Vert _{1}\leq\left\Vert P_{\sigma_{+}%
}^{M}-s\,P_{\sigma_{-}}^{M}\right\Vert _{1},\,\forall s\geq0.
\]
Also, by Lemma \thinspace\ref{lem:commute}, $\left\{  \rho_{\theta}\right\}
_{\theta\in\Theta}\equiv^{c}\left\{  \sigma_{\theta}\right\}  _{\theta
\in\Theta}$ only if
\[
\left\Vert \rho_{+}-s\,\rho_{-}\right\Vert _{1}=\left\Vert \sigma
_{+}-s\,\sigma_{-}\right\Vert _{1},\,\forall s\geq0.
\]
Therefore, we have
\[
\left\Vert \sigma_{+}-s\,\sigma_{-}\right\Vert _{1}\leq\left\Vert
P_{\sigma_{+}}^{M}-s\,P_{\sigma_{-}}^{M}\right\Vert _{1},\,\forall s\geq0.
\]
Therefore, \ by the monotonicity of $\left\Vert \cdot\right\Vert _{1}$, there
is a measurement $M$ such that
\[
\left\Vert \sigma_{+}-s\,\sigma_{-}\right\Vert _{1}=\left\Vert P_{\sigma_{+}%
}^{M}-s\,P_{\sigma_{-}}^{M}\right\Vert _{1},\,\forall s\geq0.
\]
Observe the above identity holds if and only if $M=\left\{  M_{+}%
,M_{-}\right\}  $, where $M_{+}$ and $M_{-}$ are the projector onto the
positive and the negative eigenvector of $\sigma_{+}-s\,\sigma_{-}$,
respectively. Since $M$ does not depends on $s$, combined with the fact that
$\sigma_{\theta}$ is a qubit state, \ we have
\[
\left[  \sigma_{+}-s\,\sigma_{-},\sigma_{+}-s^{\prime}\,\sigma_{-}\right]
=0,
\]
or equivalently, $\left[  \sigma_{+},\sigma_{-}\right]  =0$.
\end{proof}

\subsection{Comparison of input states}

Consider a family $\left\{  \Lambda_{\theta}\right\}  _{\theta\in\Theta}$ of
channels $\Lambda_{\theta}:\mathcal{B}\left(  \mathcal{H}_{in}\right)
\rightarrow\mathcal{B}\left(  \mathcal{H}_{out}\right)  $. We say the input
state $\rho$ is \textit{universally better} than $\rho^{\prime}$ and write
$\rho\succeq^{c}\rho^{\prime}$ if and only if $\rho$ is better than
$\rho^{\prime}$ for any statistical decision problem on $\left\{
\Lambda_{\theta}\right\}  _{\theta\in\Theta}$. More formally, $\rho\succeq
^{c}\rho^{\prime}$ if and only if
\[
\left\{  \left(  \Lambda_{\theta}\otimes\mathbf{I}\right)  \left(
\rho\right)  \right\}  _{\theta\in\Theta}\succeq_{c}\left\{  \left(
\Lambda_{\theta}\otimes\mathbf{I}\right)  \left(  \rho^{\prime}\right)
\right\}  _{\theta\in\Theta}.
\]
If $\rho\succeq^{c}\rho^{\prime}$ and $\rho^{\prime}\not \succeq ^{c}\rho$
holds, we say $\rho$ is strictly universally better than $\rho^{\prime}$, and
write $\rho\succ^{c}\rho^{\prime}$. If $\rho\succeq^{c}\rho^{\prime}$ and
$\rho^{\prime}\succeq^{c}\rho$ holds, we write $\rho\equiv^{c}\rho^{\prime}$
and say that $\rho$ and $\rho^{\prime}$ are universally equivalent.
Obviously,$\ $
\[
\rho\equiv^{c}\Lambda_{\mathbf{1}\otimes U}\left(  \rho\right)
\]
for any $U\in\mathrm{SU}\left(  \mathcal{H}_{R}\right)  $.

Denote
\[
R\left(  l,M,\pi,\rho\right)  :=\int_{\Theta\times\mathcal{D}}l_{\theta
}\left(  t\right)  \,\mathrm{d}P_{\Lambda_{\theta}\otimes\mathbf{I}\left(
\rho\right)  }^{M}\left(  t\right)  \mathrm{d}\pi\left(  \mathrm{\,}%
\theta\right)  .
\]
An input state $\rho\in\mathcal{B}\left(  \mathcal{H}_{in}\otimes
\mathcal{H}_{R}\right)  $ is said to be \textit{admissible} if and only if,
for a decision space $\mathcal{D}$ equipped with a $\sigma$-field
$\mathfrak{A}$, a $\sigma$-field $\mathfrak{B}$ over $\Theta$, a loss function
$l:\Theta\times\mathcal{D}\rightarrow%
\mathbb{R}
_{+}$ which is jointly measurable, a probability measure $\pi$ over $\left(
\Theta,\mathfrak{B}\right)  $,
\begin{equation}
\inf_{M}R\left(  l,M,\pi,\rho\right)  \leq\inf_{M}R\left(  l,M,\pi
,\rho^{\prime}\right)  ,\forall\rho.\, \label{R<R}%
\end{equation}
When the inequality in (\ref{R<R}) is strict inequality "$<$", $\rho$ is said
to be \textit{strictly admissible}.

\section{A pair of unitary operations}

\label{sec:2-unitary}

Let \ $\Lambda_{\theta}=\Upsilon_{U_{\theta}}$, $\Theta=\left\{  +,-\right\}
$ and $U_{+}$, $U_{-}\in\mathrm{SU}\left(  d\right)  $. \cite{ZimanSedlak} had
discussed discrimination $U_{+}$, $U_{-}$ and computed Bayesian error
probability and error probability of unambiguous discrimination. After
performing optimization for each case, they found that optimal input states
are minimizers of the functional
\begin{equation}
\left\vert \psi\right\rangle \rightarrow\left\vert \left\langle \psi
\right\vert U_{+}^{\dagger}U_{-}\otimes\mathbf{1}\left\vert \psi\right\rangle
\right\vert . \label{inner-prod}%
\end{equation}

Indeed, generalizing their result, we can conclude that minimizers of
(\ref{inner-prod}) are universally optimal, or optimal for any statistical
inference made upon $\left\{  \Lambda_{\theta}\right\}  _{\theta\in\Theta}$ ,
e.g., statistical test of Neyman-Pearson test, or minimax error probability.
This is an immediate consequence of Proposition\thinspace\ref{prop:randomize}
and Lemma\thinspace\ref{lem:alberti-uhlmann}.

\section{Covariant and contravariant channels}

\label{sec:covariant}

\subsection{Universally optimal input states}

Let $g\in\mathcal{G}$, where $\mathcal{G}$ is an element of compact Lie group
or its discrete subgroup. Covariant and contravariant channels are those satisfying%

\[
\Lambda_{\theta}\circ \Upsilon_{U_{g}}=\Upsilon_{V_{g}}\circ\Lambda_{\theta},
\]
and%
\[
\Lambda_{\theta}\circ \Upsilon_{U_{g}}=\Upsilon_{\overline{V_{g}}}\circ
\Lambda_{\theta},
\]
respectively. Here $g\rightarrow U_{g}$, $g\rightarrow V_{g}$ are
representations of $\mathcal{G}$.

\begin{example}
Let
\[
\Lambda_{\theta}^{\mathrm{cdep}}:=\theta T\mathbf{+}\left(  1-\theta\right)
\Upsilon_{m},
\]
where $T\left(  \rho\right)  =\rho^{T}$, $\Upsilon_{m}$ is the channel which
sends any input to the totally mixed state $\mathbf{1}/d$, and $\Theta
:=\left[  0,1/\left(  d+1\right)  \right]  \subset%
\mathbb{R}
$. Then $\Lambda_{\theta}^{\mathrm{cdep}}$ is completely positive, trace
preserving, and contravariant. \ 
\end{example}

\begin{example}
Let $\Upsilon_{c}$ be the $m$ to $n$ optimal pure state cloner\thinspace
\thinspace\cite{KeylWerner}, which is covariant with $U_{g}:=g^{\otimes m}$,
$V_{g}:=g^{\otimes n}$, $\mathcal{H}_{in}:=\left(
\mathbb{C}
^{d}\right)  ^{\otimes_{s}m}$, and $\mathcal{H}_{out}:=\left(
\mathbb{C}
^{d}\right)  ^{\otimes_{s}n}$. (Here, $\otimes_{s}$ denotes symmetric tensor
product. )Then, the channels
\[
\Lambda_{\theta}^{\mathrm{cl}}:=\theta \Upsilon_{c}+\left(  1-\theta\right)
\Upsilon_{m},\,\,\theta\in\Theta:=\left[  0,1\right]  ,
\]
are covariant.
\end{example}

\begin{example}
\label{ex:gdep}Another example is $\Lambda_{d,\theta}^{\mathrm{gp}}$ with
$\mathcal{H}_{in}:=%
\mathbb{C}
^{d}$, and $\mathcal{H}_{out}:=%
\mathbb{C}
^{d}$,
\[
\Lambda_{d,\theta}^{\mathrm{gp}}:=\sum_{j,k=0}^{d-1}\theta^{\left(
j,k\right)  }\Upsilon_{X_{d}^{j}Z_{d}^{k}},
\]
where
\[
\Theta:=\left\{  \theta\,;\,\theta^{\left(  j,k\right)  }\geq0,\,\sum
_{j,k=0}^{d-1}\theta^{\left(  j,k\right)  }=1\right\}  ,
\]
and $X_{d}$, $\,Z_{d}$ are generalized Pauli matrices defined by
\begin{equation}
X_{d}:=\sum_{i=1}^{d-1}\left\vert i\right\rangle \left\langle i+1\right\vert
+\left\vert d\right\rangle \left\langle 1\right\vert ,\,Z_{d}:=\sum_{i=1}%
^{d}e^{\frac{\sqrt{-1}2\pi i}{d}}\left\vert i\right\rangle \left\langle
i\right\vert . \label{g-pauli}%
\end{equation}
$X_{d}$ and $Z_{d}$ satisfy
\begin{equation}
\left(  X_{d}\right)  ^{d}=\left(  Z_{d}\right)  ^{d}=\mathbf{1,\,\,}%
e^{\frac{\sqrt{-1}2\pi i}{d}}Z_{d}X_{d}=X_{d}Z_{d}. \label{g-pauli-property}%
\end{equation}

\end{example}

$\Lambda_{d,\theta}^{\mathrm{gp}}$ is covariant with respect to
\[
\mathcal{G}=\mathcal{G}_{d}:=\left\{  e^{\frac{\sqrt{-1}\,2\pi\,i}{d}%
}\,\left(  X_{d}\right)  ^{j}\,\left(  Z_{d}\right)  ^{k};i,j,k=0,1,\cdots
,d-1\right\}
\]
\ and $U_{g}=V_{g}=g$. Indeed, if $\mathcal{H}_{in}=\mathcal{H}_{out}=%
\mathbb{C}
^{d}$ and $U_{g}=V_{g}=g\in\mathcal{G}_{d}$, being covariant is equivalent to
be a member of $\left\{  \Lambda_{d,\theta}^{\mathrm{gp}}\right\}  $%
\thinspace\cite{Matsumoto}.

\begin{example}
An alternative parameterization of $\Lambda_{2,\theta}^{\mathrm{gp}}$ is given
by
\[
\Lambda_{2,\eta}^{\mathrm{gp}}:=\sum_{i=1}^{4}\Upsilon_{E_{i}},
\]
where
\begin{align*}
E_{1}  &  :=\left[
\begin{array}
[c]{cc}%
\eta^{1} & 0\\
0 & \eta^{2}%
\end{array}
\right]  ,\,E_{2}:=\left[
\begin{array}
[c]{cc}%
\eta^{2} & 0\\
0 & \eta^{1}%
\end{array}
\right]  ,\\
E_{3}  &  :=\left[
\begin{array}
[c]{cc}%
0 & \eta^{3}\\
\sqrt{1-\sum_{i=1}^{3}\left(  \eta^{i}\right)  ^{2}} & 0
\end{array}
\right]  ,\\
E_{4}  &  :=E_{3}^{\dagger}.
\end{align*}

\end{example}

\begin{example}
\label{ex:g-damp}With $\mathcal{H}_{in}=\mathcal{H}_{out}=%
\mathbb{C}
^{2}$, \cite{Fujiwara:2004} had defined generalized damping channels :%
\begin{align*}
\Lambda_{p,\xi}^{\mathrm{damp}}  &  =\sum_{i=1}^{4}\Upsilon_{F_{i}},\\
F_{1}  &  :=\sqrt{p}\left[
\begin{array}
[c]{cc}%
1 & 0\\
0 & \sqrt{\xi}%
\end{array}
\right]  ,\,F_{2}:=\sqrt{1-p}\left[
\begin{array}
[c]{cc}%
\sqrt{\xi} & 0\\
0 & 1
\end{array}
\right]  ,\\
F_{3}  &  :=\sqrt{p}\left[
\begin{array}
[c]{cc}%
0 & \sqrt{1-\xi}\\
0 & 0
\end{array}
\right]  ,\,F_{4}:=\sqrt{1-p}\left[
\begin{array}
[c]{cc}%
0 & 0\\
\sqrt{1-\xi} & 0
\end{array}
\right]  .
\end{align*}
Then, $\left\{  \Lambda_{1/2,\xi}^{\mathrm{damp}}\right\}  $ is covariant with
respect to $U_{g}=V_{g}=g\in\mathcal{G}_{2}$. Indeed, $\left\{  \Lambda
_{1/2,\xi}^{\mathrm{damp}}\right\}  $ is a subset of $\left\{  \Lambda
_{2,\theta}^{\mathrm{gp}}\right\}  $.
\end{example}

\begin{example}
\label{ex:diag}Let
\[
\Lambda_{d,\theta}^{\mathrm{diag}}:=\sum_{i=1}^{d}\Upsilon_{E_{i}},
\]
where%
\begin{align*}
E_{1}  &  :=\mathrm{diag}\left(  \theta^{1},\theta^{2},\cdots,\theta
^{\left\lceil \left(  d-1\right)  /2\right\rceil },\sqrt{\sum_{i=1}%
^{\left\lceil \left(  d-1\right)  /2\right\rceil }\left(  \theta^{i}\right)
^{2}},0,\cdots,0\right)  ,\\
E_{i}  &  :=X_{d}^{i-1}E_{1}X_{d}^{i-1}\,\,(2\leq i\leq d).
\end{align*}
Then $\Lambda_{d,\xi}^{\mathrm{diag}}$ is covariant with respect to
$U_{g}=V_{g}=g\in\mathcal{G}_{d}$: in fact, it turns out the family $\left\{
\Lambda_{d,\xi}^{\mathrm{diag}}\right\}  $ is a subset of the family $\left\{
\Lambda_{\theta}^{\mathrm{gp}}\right\}  $.
\end{example}

\bigskip

\cite{FujiwaraImai:2003} and \cite{Fujiwara:2004} had shown that a maximal
entangles state $\left\vert \Phi_{d}\right\rangle $ is optimal for the family
$\left\{  \Lambda_{2,\theta}^{\mathrm{gp}}\right\}  $ \ and also for the
family $\left\{  \Lambda_{1/2,\xi}^{\mathrm{damp}}\right\}  $ in the following
sense. For any input state $\rho_{in}\in\mathcal{B}\left(  \mathcal{H}%
_{in}\otimes\mathcal{H}_{R}\right)  $ and a measurement $M$ over
$\mathcal{B}\left(  \mathcal{H}_{out}\otimes\mathcal{H}_{R}\right)  $ with
\begin{equation}
\mathrm{E}\left[  M,\Lambda_{\theta}\otimes\mathbf{I}\left(  \rho_{in}\right)
\right]  =\theta, \label{EML=0}%
\end{equation}
there is a measurement $M^{\prime}$ such that%
\[
\mathrm{E}\left[  M^{\prime},\Lambda_{\theta}\otimes\mathbf{I}\left(
\left\vert \Phi_{d}\right\rangle \left\langle \Phi_{d}\right\vert \right)
\right]  =\theta
\]
and
\[
\mathrm{V}\left[  M,\Lambda_{\theta}\otimes\mathbf{I}\left(  \rho_{in}\right)
\right]  =\mathrm{V}\left[  M^{\prime},\Lambda_{\theta}\otimes\mathbf{I}%
\left(  \left\vert \Phi_{d}\right\rangle \left\langle \Phi_{d}\right\vert
\right)  \right]  .
\]

Also, \cite{Sacchi:2005} studies discrimination of a pair of channels in
$\left\{  \Lambda_{d,\theta}^{\mathrm{gp}}\right\}  $, and shows that
$\left\vert \Phi_{d}\right\rangle $ minimizes Bayesian error probability for
any prior distributions. In case of qubits, they extended their result to
minimax error probability\thinspace\cite{Sacchi:2005:3}.

The following theorem is a generalization of these results. Below, we suppose
the representation $g\rightarrow U_{g}$ occurs the decomposition
\[
\mathcal{H}_{in}=\bigoplus_{\mu}\mathcal{H}_{in}^{\left(  \mu\right)
},\,\,U_{g}=\bigoplus_{\mu}U_{g}^{\left(  \mu\right)  },
\]
where $U_{g}^{\left(  \mu\right)  }$ acts on $\mathcal{H}_{in}^{\left(
\mu\right)  }$ and $g\rightarrow U_{g}^{\left(  \mu\right)  }$ is irreducible.
Also, define $d_{\mu}:=\dim\mathcal{H}_{in}^{\left(  \mu\right)  }$.

\begin{theorem}
\label{th:covariant-opt-input}Consider the covariant or contravariant channel
family $\left\{  \Lambda_{\theta}\right\}  _{\theta\in\Theta}$. Then, with
$\mathcal{H}_{R}^{\left(  \mu\right)  }\simeq\mathcal{H}_{in}^{\left(
\mu\right)  }$ and\
\[
\mathcal{H}_{R}=\bigoplus_{\mu}\mathcal{H}_{R}^{\left(  \mu\right)  },
\]
the followidng input state is universally optimal:
\begin{equation}
\left\vert \psi_{opt}\right\rangle :=c\bigoplus_{\mu}\left\vert \Phi_{d_{\mu}%
}\right\rangle , \label{covariant-opt}%
\end{equation}
where $\left\vert \Phi_{d_{\mu}}\right\rangle \in\mathcal{H}_{in}^{\left(
\mu\right)  }\otimes\mathcal{H}_{R}^{\left(  \mu\right)  }$, and $c$ is the
normalizing constant.
\end{theorem}

\begin{proof}
We state the proof only for covariant case, since the argument is almost
parallel for contravariant case. Below, we compose a completely positive trace
preserving map $\Gamma_{\psi}$ with
\[
\Gamma_{\psi}\left(  \Lambda_{\theta}\otimes\mathbf{I}\left(  \left\vert
\psi_{opt}\right\rangle \left\langle \psi_{opt}\right\vert \right)  \right)
=\Lambda_{\theta}\otimes\mathbf{I}\left(  \left\vert \psi\right\rangle
\left\langle \psi\right\vert \right)  \text{,}%
\]
for an arbitrary $\left\vert \psi\right\rangle \in\mathcal{H}_{in}^{\prime
}\otimes\mathcal{H}_{R}^{\prime}$, and use Proposition\thinspace
\ref{prop:randomize}. Here,
\begin{align*}
\Lambda_{\theta}\otimes\mathbf{I}\left(  \left\vert \psi_{opt}\right\rangle
\left\langle \psi_{opt}\right\vert \right)   &  \in\mathcal{H}_{out}%
\otimes\mathcal{H}_{R},\\
\Lambda_{\theta}\otimes\mathbf{I}\left(  \left\vert \psi\right\rangle
\left\langle \psi\right\vert \right)   &  \in\mathcal{H}_{out}\otimes
\mathcal{H}_{R}^{\prime},
\end{align*}
where $\mathcal{H}_{R}^{\prime}\simeq\mathcal{H}_{R}$.

$\Gamma_{\psi}$ is composed as follows; Prepare $\left\vert \psi\right\rangle
$ in $\mathcal{H}_{in}^{\prime}\otimes\mathcal{H}_{R}^{\prime}$, where
$\mathcal{H}_{in}^{\prime}\simeq\mathcal{H}_{in}$. Apply the measurement $M$
(defined later) jointly to $\mathcal{H}_{R}$-part of $\Lambda_{\theta}%
\otimes\mathbf{I}\left(  \left\vert \psi_{opt}\right\rangle \left\langle
\psi_{opt}\right\vert \right)  $ and $\mathcal{H}_{in}^{\prime}$-part of
$\left\vert \psi\right\rangle $. Depending on the outcome $g\in\mathcal{G}$ of
$M$, apply $V_{g}^{\dagger}$ to $\mathcal{H}_{out}$.

To define the measurement $M$, we first define the the state vector in
$\mathcal{H}_{R}\otimes\mathcal{H}_{in}^{\prime}$,
\[
\left\vert \varphi_{g}\right\rangle :=c^{\prime}\bigoplus_{\mu}d_{\mu
}\overline{U_{g}^{\left(  \mu\right)  }}\otimes\mathbf{1}_{\mathcal{H}%
_{in}^{\left(  \mu\right)  \prime}}\left\vert \Phi_{d_{\mu}}\right\rangle ,
\]
with the normalizing constant $c^{\prime}$, $U_{g}^{\left(  \mu\right)  }$
being in $\mathcal{H}_{R}^{\left(  \mu\right)  }$\thinspace and $\mathcal{H}%
_{in}^{\left(  \mu\right)  \prime}\backsimeq\mathcal{H}_{in}^{\left(
\mu\right)  }$. Then, the measurement $M$ is the one which occurs state
change
\[
\rho\rightarrow c^{\prime\prime}\mathbf{I}_{\mathcal{H}_{out}\otimes
\mathcal{H}_{R}^{\prime}}\otimes \Upsilon_{\left\langle \varphi_{g}\right\vert
}\left(  \rho\right)  ,\,
\]
with the probability density $\mathrm{tr}\,\mathbf{I}_{\mathcal{H}%
_{out}\otimes\mathcal{H}_{R}^{\prime}}\otimes \Upsilon_{\left\langle
\varphi_{g}\right\vert }\left(  \rho\right)  $. Here $c^{\prime\prime}$ is the
normalizing constant, and the density is considered with respect to the Haar
measure $\mathrm{d}g$ such that $\int_{\mathcal{G}}\mathrm{d}g=1$.

In the end, we confirm that $\Gamma_{\psi}$ meets the requirement. By
composition, $\Gamma_{\psi}$ is completely positive and trace preserving. Let
$\left\{  A_{\kappa}\right\}  $ be the Kraus operators of $\Lambda_{\theta}$.
Also, let $\mathcal{H}_{R}^{\prime}:=\bigoplus_{\mu}\mathcal{H}_{R}^{\left(
\mu\right)  \prime}$ and $\mathcal{H}_{R}^{\left(  \mu\right)  \prime
}\backsimeq\mathcal{H}_{R}^{\left(  \mu\right)  }$. Then, after the
application of $M$ and obtaining measurement result $g\in\mathcal{G}$, the
state will be the mixture of the pure state in $\mathcal{H}_{out}%
\otimes\mathcal{H}_{R}^{\prime}$, such that
\begin{align*}
&  \sqrt{c^{\prime\prime}}\left(  \mathbf{1}_{\mathcal{H}_{out}}%
\otimes\left\langle \varphi_{g}\right\vert \otimes\mathbf{1}_{\mathcal{H}%
_{R}^{\prime}}\right)  \left(  A_{\kappa}\otimes\mathbf{1}_{\mathcal{H}_{R}%
}\otimes\mathbf{1}_{\mathcal{H}_{in}^{\prime}}\otimes\mathbf{1}_{\mathcal{H}%
_{R}^{\prime}}\left\vert \psi_{opt}\right\rangle \left\vert \psi\right\rangle
\right) \\
&  =\sqrt{c^{\prime\prime}}\left(  A_{\kappa}\otimes\mathbf{1}_{\mathcal{H}%
_{R}^{\prime}}\right)  \left(  \mathbf{1}_{\mathcal{H}_{in}}\otimes
\left\langle \varphi_{g}\right\vert \otimes\mathbf{1}_{\mathcal{H}_{R}%
^{\prime}}\right)  \left\vert \psi_{opt}\right\rangle \left\vert
\psi\right\rangle \\
&  =cc^{\prime}\sqrt{c^{\prime\prime}}\left(  A_{\kappa}\otimes\mathbf{1}%
_{\mathcal{H}_{R}^{\prime}}\right)  \bigoplus_{\mu}d_{\mu}\left(
\mathbf{1}_{\mathcal{H}_{in}}\otimes\left(  \left\langle \Phi_{d_{\mu}%
}\right\vert U_{g}^{\left(  \mu\right)  T}\otimes\mathbf{1}_{\mathcal{H}%
_{in}^{\prime}}\right)  \otimes\mathbf{1}_{\mathcal{H}_{R}^{\left(
\mu\right)  \prime}}\right)  \left\vert \Phi_{d_{\mu}}\right\rangle \left\vert
\psi\right\rangle \\
&  =cc^{\prime}\sqrt{c^{\prime\prime}}\left(  A_{\kappa}\otimes\mathbf{1}%
_{\mathcal{H}_{R}^{\prime}}\right)  \bigoplus_{\mu}d_{\mu}\left(
\mathbf{1}_{\mathcal{H}_{in}}\otimes\left(  \left\langle \Phi_{d_{\mu}%
}\right\vert \mathbf{1}_{\mathcal{H}_{R}^{\left(  \mu\right)  }}\otimes
U_{g}^{\left(  \mu\right)  }\right)  \otimes\mathbf{1}_{\mathcal{H}%
_{R}^{\left(  \mu\right)  \prime}}\right)  \left\vert \Phi_{d_{\mu}%
}\right\rangle \left\vert \psi\right\rangle \\
&  =cc^{\prime}\sqrt{c^{\prime\prime}}A_{\kappa}\bigoplus_{\mu}\sum
_{i=1}^{d_{\mu}}\left\vert i\right\rangle _{\mathcal{H}_{in}^{\left(
\mu\right)  }}\otimes\left(  \left(  _{\mathcal{H}_{in}^{\left(  \mu\right)
\prime}}\,\left\langle i\right\vert U_{g}^{\left(  \mu\right)  }\right)
\otimes\mathbf{1}_{\mathcal{H}_{R}^{\left(  \mu\right)  \prime}}\left\vert
\psi\right\rangle \right) \\
&  =cc^{\prime}\sqrt{c^{\prime\prime}}\left(  A_{\kappa}U_{g}\right)
\otimes\mathbf{1}_{\mathcal{H}_{R}^{\prime}}\left\vert \psi\right\rangle .
\end{align*}
This mixture equals
\begin{align*}
&  \left(  cc^{\prime}\right)  ^{2}c^{\prime\prime}\left(  \Lambda_{\theta
}\circ \Upsilon_{U_{g}}\right)  \otimes\mathbf{I}_{\mathcal{H}_{R}^{\prime}%
}\left(  \left\vert \psi\right\rangle \left\langle \psi\right\vert \right) \\
&  =\left(  cc^{\prime}\right)  ^{2}c^{\prime\prime}\left(  \Upsilon_{V_{g}%
}\circ\Lambda_{\theta}\right)  \otimes\mathbf{I}_{\mathcal{H}_{R}^{\prime}%
}\left(  \left\vert \psi\right\rangle \left\langle \psi\right\vert \right)  .
\end{align*}
Therefore, applying $V_{g}^{\dagger}$ $\ $to $\mathcal{H}_{out}$, we have
$\Lambda_{\theta}\otimes\mathbf{I}_{\mathcal{H}_{R}^{\prime}}\left(
\left\vert \psi\right\rangle \left\langle \psi\right\vert \right)  $, as desired.
\end{proof}

\subsection{On $\Lambda_{1,\xi}^{\mathrm{damp}}$}

\cite{Fujiwara:2004} had shown that for $\Lambda_{1,\xi}^{\mathrm{damp}}$,
$\rho=\left\vert 2\right\rangle \left\langle 2\right\vert \in\mathcal{B}%
\left(  \mathcal{H}_{in}\right)  $ is optimal for mean square error under the
constraint (\ref{EML=0}). Despite this fact, $\left\vert 2\right\rangle
\left\langle 2\right\vert $ is \textit{not} universally optimal as is shown
below. Indeed, \ %

\begin{align*}
&  \left\Vert \Lambda_{1,\xi}^{\mathrm{damp}}\left(  \left\vert 2\right\rangle
\left\langle 2\right\vert \right)  -s\,\Lambda_{1,0}^{\mathrm{damp}}\left(
\left\vert 2\right\rangle \left\langle 2\right\vert \right)  \right\Vert
_{1}=\left\vert 1-\xi-s\right\vert +\xi,\\
&  \left\Vert \Lambda_{1,\xi}^{\mathrm{damp}}\otimes\mathbf{I}\left(
\left\vert \Phi_{2}\right\rangle \left\langle \Phi_{2}\right\vert \right)
-s\,\Lambda_{1,0}^{\mathrm{damp}}\otimes\mathbf{I}\left(  \left\vert \Phi
_{2}\right\rangle \left\langle \Phi_{2}\right\vert \right)  \right\Vert _{1}\\
&  =\frac{1}{2}\sqrt{\left(  1-s+\xi\right)  ^{2}+4s\xi}+\frac{1}{2}\left\vert
1-\xi-s\right\vert .
\end{align*}
Therefore,
\begin{align*}
\frac{1}{2}  &  =\left\Vert \Lambda_{1,1/2}^{\mathrm{damp}}\left(  \left\vert
2\right\rangle \left\langle 2\right\vert \right)  -\frac{1}{2}\,\Lambda
_{1,0}^{\mathrm{damp}}\left(  \left\vert 2\right\rangle \left\langle
2\right\vert \right)  \right\Vert _{1}\\
&  <\left\Vert \Lambda_{1,1/2}^{\mathrm{damp}}\otimes\mathbf{I}\left(
\left\vert \Phi_{2}\right\rangle \left\langle \Phi_{2}\right\vert \right)
-\frac{1}{2}\,\Lambda_{1,0}^{\mathrm{damp}}\otimes\mathbf{I}\left(  \left\vert
\Phi_{2}\right\rangle \left\langle \Phi_{2}\right\vert \right)  \right\Vert
_{1}=\frac{\sqrt{2}}{2}.
\end{align*}
Therefore, by Lemma\thinspace\ref{lem:commute}, we have the assertion.

\subsection{An alternative proof for $\left\{  \Lambda_{d,\theta}%
^{\mathrm{gp}}\right\}  _{\theta\in\Theta}$}

Given $\Lambda_{d,\theta}^{\mathrm{gp}}\otimes\mathbf{I}\left(  \left\vert
\Phi_{d}\right\rangle \left\langle \Phi_{d}\right\vert \right)  $ and
$\rho_{in}\in\mathcal{B}\left(  \mathcal{H}\right)  $, one can generate
$\Lambda_{d,\theta}^{\mathrm{gp}}\otimes\mathbf{I}\left(  \rho_{in}\right)  $
in the following manner. Measure $\Lambda_{\theta}\otimes\mathbf{I}\left(
\left\vert \Phi_{d}\right\rangle \left\langle \Phi_{d}\right\vert \right)  $
by the projectors onto $\left\{  X_{d}^{j}Z_{d}^{k}\otimes\mathbf{1}\left\vert
\Phi_{d}\right\rangle \right\}  _{j,k=0}^{d-1}$, and apply the unitary
$X_{d}^{j}Z_{d}^{k}\otimes\mathbf{1}$ if $\left(  j,k\right)  $ is observed.

This composition works also for any channel family $\left\{  \Lambda_{\theta
}^{\mathrm{ou}}\right\}  $ with%
\begin{align*}
\Lambda_{\theta}^{\mathrm{ou}}  &  :=\frac{1}{d}\left(  1-\sum_{i=1}^{d^{2}%
-1}\theta^{i}\right)  \Upsilon_{U_{1}}+\sum_{i=2}^{d^{2}}\theta^{i}%
\Upsilon_{U_{i}},\\
\mathrm{tr}\,U_{i}U_{j}^{\dagger}  &  =d\delta_{ij}.
\end{align*}

\section{Unital qubit channels}

\label{sec:d=2}

In this section, $\mathcal{H}_{in}=\mathcal{H}_{out}=%
\mathbb{C}
^{2}$. Also we denote
\[
Y_{2}:=\sqrt{-1}Z_{2}X_{2}.
\]
and define
\[
V=\left[
\begin{array}
[c]{cc}%
e^{\sqrt{-1}b}\cos a & -e^{-\sqrt{-1}c}\sin a\\
e^{\sqrt{-1}c}\sin a & e^{-\sqrt{-1}b}\cos a
\end{array}
\right]  \in\mathrm{SU}\left(  2\right)  .
\]
With $\boldsymbol{p}=\left(  p_{1},p_{2}\right)  $ ($p_{1}+p_{2}=1$) and
$V\in\mathrm{SU}\left(  2\right)  $, let
\[
\left\vert \varphi_{\boldsymbol{p},V}\right\rangle :=\sqrt{p_{1}}\left(
V\left\vert 1\right\rangle \right)  \otimes\left\vert 1\right\rangle
+\sqrt{p_{2}}\left(  V\left\vert 2\right\rangle \right)  \otimes\left\vert
2\right\rangle .
\]
Also, let $\Gamma_{UNOT}$ denote the universal not operation
\[
\Gamma_{\mathrm{unot}}\left(  C\right)  =\overline{\Lambda_{Y_{2}}\left(
C\right)  },
\]
which is positive trace preserving but not completely positive.

Observe that $Y_{2}$ is unitary and Hermite, and that
\begin{equation}
Y_{2}\,V\,=\overline{V}\,Y_{2} \label{ZXV=VZX}%
\end{equation}
or equivalently,
\begin{equation}
\Upsilon_{Y_{2}}\circ\Gamma_{\mathrm{unot}}=\Gamma_{\mathrm{unot}}%
\circ \Upsilon_{Y_{2}}. \label{unot-commute}%
\end{equation}

\begin{lemma}
Suppose%
\begin{equation}
Y_{2}\,\Lambda_{\theta}\left(  C\right)  \,Y_{2}=\overline{\Lambda}_{\theta
}\left(  Y_{2}\,C\,Y_{2}\right)  \,, \label{symmetry}%
\end{equation}
or equivalently
\begin{equation}
\Lambda_{\theta}\circ\Gamma_{\mathrm{unot}}=\Gamma_{\mathrm{unot}}\circ
\Lambda_{\theta} \label{unot-commute-2}%
\end{equation}
Then, the input $\left\vert \Phi_{2}\right\rangle $ is universally optimal.
\end{lemma}

\begin{proof}
To use Proposition\thinspace\ref{prop:randomize}, we compose a trace
preserving positive map $\Gamma$ with%
\[
\Gamma\left(  \Lambda_{\theta}\otimes\mathbf{I}\left(  \left\vert \Phi
_{2}\right\rangle \left\langle \Phi_{2}\right\vert \right)  \right)
=\Lambda_{\theta}\otimes\mathbf{I}\left(  \left\vert \varphi_{\boldsymbol{p}%
,V}\right\rangle \left\langle \varphi_{\boldsymbol{p},V}\right\vert \right)
\]
as follows. First, apply the unitary $V^{T}$ to $\mathcal{H}_{R}$-part of
$\Lambda_{\theta}\otimes\mathbf{I}\left(  \left\vert \Phi_{2}\right\rangle
\left\langle \Phi_{2}\right\vert \right)  $, obtaining
\begin{align*}
&  \Lambda_{\theta}\otimes\mathbf{I}\left(  \left(  \mathbf{1}\otimes
V^{T}\right)  \left\vert \Phi_{2}\right\rangle \left\langle \Phi
_{2}\right\vert \left(  \mathbf{1}\otimes V^{T\dagger}\right)  \right) \\
&  =\Lambda_{\theta}\otimes\mathbf{I}\left(  \left(  V\otimes\mathbf{1}%
\right)  \left\vert \Phi_{2}\right\rangle \left\langle \Phi_{2}\right\vert
\left(  V^{\dagger}\otimes\mathbf{1}\right)  \right)  .
\end{align*}
Second, measure $\mathcal{H}_{R}$-part by the measurement specified by the
instrument
\[
\left\{  \sqrt{M},\sqrt{\mathbf{1}-M}\right\}  ,
\]
where
\[
M:=p_{1}\left\vert 1\right\rangle \left\langle 1\right\vert +p_{2}\left\vert
2\right\rangle \left\langle 2\right\vert .
\]
\ 

If the measurement result is the one corresponding to $\sqrt{M}$, then we are
done. Otherwise, letting\ $\boldsymbol{p}^{\prime}:=\left(  p_{2}%
,p_{1}\right)  $, we obtain%
\begin{align*}
&  \Lambda_{\theta}\otimes\mathbf{I}\left(  \left\vert \varphi_{\boldsymbol{p}%
^{\prime},V}\right\rangle \left\langle \varphi_{\boldsymbol{p}^{\prime}%
,V}\right\vert \right) \\
&  =\left(  \Lambda_{\theta}\otimes\mathbf{I}\right)  \circ\left(
\Gamma_{\mathrm{unot}}\otimes\Gamma_{\mathrm{unot}}\right)  \left(  \left\vert
\varphi_{\boldsymbol{p},V}\right\rangle \left\langle \varphi_{\boldsymbol{p}%
,V}\right\vert \right) \\
&  =\left(  \Gamma_{\mathrm{unot}}\otimes\Gamma_{\mathrm{unot}}\right)
\circ\left(  \Lambda_{\theta}\otimes\mathbf{I}\right)  \left(  \left\vert
\varphi_{\boldsymbol{p},V}\right\rangle \left\langle \varphi_{\boldsymbol{p}%
,V}\right\vert \right)  .
\end{align*}
\ So we apply $\Gamma_{\mathrm{unot}}\otimes\Gamma_{\mathrm{unot}}$ , to
obtain $\Lambda_{\theta}\otimes\mathbf{I}\left(  \left\vert \varphi
_{\boldsymbol{p},V}\right\rangle \left\langle \varphi_{\boldsymbol{p}%
,V}\right\vert \right)  $.
\end{proof}

Any $4\times4$ Hermite matrix belongs to%

\[
\mathrm{span}_{%
\mathbb{R}
}\left\{  A\otimes B\,;\,A,B=\mathbf{1},X_{2},Y_{2},Z_{2}\right\}  .
\]
So is Choi-Jamilokovski's representation $Ch\left(  \Lambda_{\theta}\right)
$. Since $\Lambda_{\theta}$ is trace preserving,
\[
\mathrm{tr}\,_{\mathcal{H}_{out}\,}Ch\left(  \Lambda_{\theta}\right)
=\mathbf{1}_{in}.
\]
Therefore, $Ch\left(  \Lambda_{\theta}\right)  $ is a positive element of
$\mathrm{span}_{%
\mathbb{R}
}\,\mathcal{TP}$, where
\[
\mathcal{TP}:=\left\{  A\otimes B\,;A,B=\mathbf{1},X_{2},Y_{2},Z_{2},\text{ if
}A=\mathbf{1}\text{, then }B=\mathbf{1}\right\}
\]

\begin{lemma}
(\ref{unot-commute-2}) holds if $Ch\left(  \Lambda_{\theta}\right)  $ is an
element of
\[
\mathrm{span}_{%
\mathbb{R}
}\left(  \mathcal{TP\,}-\left\{  X_{2}\otimes\mathbf{1,}Y_{2}\otimes
\mathbf{1,}Z_{2}\otimes\mathbf{1}\right\}  \right)  ,
\]
or equivalently,
\[
\Lambda_{\theta}\left(  \mathbf{1}\right)  =\mathbf{1}.
\]

\end{lemma}

\begin{proof}
Since
\[
\Lambda_{\theta}\left(  C\right)  =\mathrm{tr}_{\mathcal{H}_{in}\,}\,Ch\left(
\Lambda_{\theta}\right)  \left(  \mathbf{1}_{out}\otimes C^{T}\right)  ,
\]
and
\begin{align*}
\overline{\Lambda_{\theta}}\left(  ACA^{\dagger}\right)   &  =\mathrm{tr}%
_{\mathcal{H}_{in}}\,\left(  \mathbf{1}\otimes A^{T}\overline{Ch\left(
\Lambda_{\theta}\right)  }\mathbf{1}\otimes\overline{A}\right)  \left(
\mathbf{1}\otimes C^{T}\right)  ,\\
A\Lambda_{\theta}\left(  C\right)  A^{\dagger}  &  =\mathrm{tr}_{\mathcal{H}%
_{in}\,}\,\,\left(  A\otimes\mathbf{1}\,Ch\left(  \Lambda_{\theta}\right)
\,A^{\dagger}\otimes\mathbf{1}\right)  \left(  \mathbf{1}\otimes C^{T}\right)
,
\end{align*}
(\ref{symmetry}) is equivalent to
\[
\mathbf{1}\otimes Y_{2}\,\,\overline{Ch\left(  \Lambda_{\theta}\right)
}\,\mathbf{1}\otimes Y_{2}=Y_{2}\otimes\mathbf{1}\,Ch\left(  \Lambda_{\theta
}\right)  Y_{2}\otimes\mathbf{1},
\]
or equivalently, with $\ W=Ch\left(  \Lambda_{\theta}\right)  $,
\begin{equation}
W=\Gamma_{\mathrm{unot}}\otimes\Gamma_{\mathrm{unot}}\left(  W\right)  .
\label{symmetry-2}%
\end{equation}
Each element of $\mathcal{TP}$ other than $X_{2}\otimes\mathbf{1}$,
$Y_{2}\otimes\mathbf{1}$ and $Z_{2}\otimes\mathbf{1}$ satisfies
(\ref{symmetry-2}). Therefore, we have the assertion.
\end{proof}

Combining these lemmas, we have the following theorem.

\begin{theorem}
\label{th:unital-qubit}Suppose $\mathcal{H}_{in}=\mathcal{H}_{out}=%
\mathbb{C}
^{2}$. Then, the input $\left\vert \Phi_{2}\right\rangle $ is universally
optimal if $\Lambda_{\theta}$ is unital.
\end{theorem}

\begin{example}
\label{ex:su2}Due to (\ref{unot-commute}), the family $\left\{  \Upsilon
_{U}\,;\,U\in\mathrm{SU}\left(  2\right)  \right\}  $ satisfies
(\ref{unot-commute-2}).
\end{example}

\begin{example}
\label{ex:2-diag}Channel family $\left\{  \Lambda_{\theta}\right\}  $ with%
\[
Ch\left(  \Lambda_{\theta}\right)  =\left[
\begin{array}
[c]{cccc}%
1 & 0 & 0 & \theta^{1}-\sqrt{-1}\theta^{2}\\
0 & 0 & 0 & 0\\
0 & 0 & 0 & 0\\
\theta^{1}+\sqrt{-1}\theta^{2} & 0 & 0 & 1
\end{array}
\right]
\]
satisfies (\ref{unot-commute-2}). In Kraus representation, $\Lambda_{\theta}$
is expressed as%
\begin{align*}
\Lambda_{\theta}  &  =\sum_{i=1}^{2}\Upsilon_{E_{i}},\\
E_{1}  &  =\left[
\begin{array}
[c]{cc}%
1 & 0\\
0 & \theta^{1}+\sqrt{-1}\theta^{2}%
\end{array}
\right]  ,E_{2}=\left[
\begin{array}
[c]{cc}%
0 & 0\\
0 & \sqrt{1-\sum_{i=1}^{2}\left(  \theta^{i}\right)  ^{2}}%
\end{array}
\right]  .
\end{align*}

\end{example}

\section{A pair of measurements}

\label{sec:measurement}

Let us consider a family $\left\{  \Lambda_{\theta}\right\}  _{\theta
\in\left\{  +,-\right\}  }$ such that $\Lambda_{\theta}:\mathcal{B}\left(
\mathcal{H}_{in}\right)  \rightarrow\mathcal{B}\left(  \mathcal{H}%
_{out}\right)  $, $\mathcal{H}_{out}=%
\mathbb{C}
^{m}$, and
\begin{equation}
\Lambda_{\theta}\left(  \rho\right)  =\sum_{i=1}^{m}\left\{  \mathrm{tr}\,\rho
M_{\theta}\left(  i\right)  \right\}  \left\vert i\right\rangle \left\langle
i\right\vert . \label{measurement}%
\end{equation}
This corresponds to measurements which outputs classical data $"i"\,$ with
probability $\mathrm{tr}\,\rho M_{\theta}\left(  i\right)  $.

\begin{example}
\label{ex:ortho-measure}Suppose
\begin{equation}
M_{+}\left(  i\right)  M_{-}\left(  i\right)  =0,\,\,\left(  i=1,\cdots
,m\right)  . \label{meas-ortho}%
\end{equation}
For example, suppose
\begin{align}
\mathrm{rank}M_{+}\left(  i\right)   &  =1,\nonumber\\
M_{-}\left(  i\right)   &  :=\frac{1}{d-1}\left\{  \mathrm{tr}M_{+}\left(
i\right)  \,\cdot\mathbf{1}-M_{+}\left(  i\right)  \right\}  .
\label{meas-ortho-1}%
\end{align}
Then,
\begin{align*}
M_{-}\left(  i\right)  M_{+}\left(  i\right)   &  =M_{+}\left(  i\right)
M_{-}\left(  i\right) \\
&  =\frac{1}{d-1}\left\{  \mathrm{tr}M_{+}\,\left(  i\right)  \cdot
M_{-}\left(  i\right)  -\left(  M_{-}\left(  i\right)  \right)  ^{2}\right\}
\\
&  =0
\end{align*}
and
\begin{align*}
\sum_{i=1}^{m}M_{-}\left(  i\right)   &  =\frac{1}{d-1}\left\{  \mathrm{tr}%
\sum_{i=1}^{m}M_{+}\left(  i\right)  \,\cdot\mathbf{1}-\sum_{i=1}^{m}%
M_{+}\left(  i\right)  \right\} \\
&  =\frac{1}{d-1}\left\{  \mathrm{tr}\,\mathbf{1}\cdot\mathbf{1-1}\right\}
\mathbf{=1.}%
\end{align*}
Thus, (\ref{meas-ortho-1}) is a special case of (\ref{meas-ortho}).

An input $\left\vert \psi\right\rangle $ is universally optimal if
\begin{equation}
\sqrt{\rho_{\psi}}M_{+}\left(  i\right)  ^{T}\rho_{\psi}M_{-}\left(  i\right)
^{T}\sqrt{\rho_{\psi}}=0,\,\left(  i=1,\cdots,m\right)  , \label{c=0}%
\end{equation}
where $\rho_{\psi}$ is as of (\ref{rho-psi}). In particular, $\left\vert
\Phi_{d}\right\rangle $ is universally optimal.

The proof is as follows. Suppose (\ref{c=0}) holds. Then%
\begin{align*}
&  \Lambda_{+}\otimes\mathbf{I}\left(  \left\vert \psi\right\rangle
\left\langle \psi\right\vert \right)  \Lambda_{-}\otimes\mathbf{I}\left(
\left\vert \psi\right\rangle \left\langle \psi\right\vert \right) \\
&  =\sum_{i=1}^{m}\left\vert i\right\rangle \left\langle i\right\vert
\otimes\sqrt{\rho_{\psi}}M_{+}\left(  i\right)  ^{T}\rho_{\psi}M_{-}\left(
i\right)  ^{T}\sqrt{\rho_{\psi}}\\
&  =0
\end{align*}
Therefore, \ $\Lambda_{+}\otimes\mathbf{I}\left(  \left\vert \psi\right\rangle
\left\langle \psi\right\vert \right)  $ and $\Lambda_{-}\otimes\mathbf{I}%
\left(  \left\vert \psi\right\rangle \left\langle \psi\right\vert \right)  $
can be discriminated perfectly. Therefore, for any $\rho_{in}\in
\mathcal{B}\left(  \mathcal{H}_{in}\otimes\mathcal{H}_{R}\right)  $, there is
a trace preserving CPTP map $\Gamma$ with
\[
\Gamma\left(  \Lambda_{\theta}\otimes\mathbf{I}\left(  \left\vert
\psi\right\rangle \left\langle \psi\right\vert \right)  \right)
=\Lambda_{\theta}\otimes\mathbf{I}\left(  \rho_{in}\right)  \,\,\,\,(\theta
=+,-),
\]
and by Proposition \ref{prop:randomize}, we have the assertion.
\end{example}

\begin{example}
\label{ex:mes-rotate}Let $\left\{  U_{i}\right\}  _{i=1}^{m}$ be unitary
matrices such that
\begin{equation}
\sum_{i=1}^{m}U_{i}AU_{i}^{\dagger}=c\,\left(  \mathrm{tr}\,A\right)
\mathbf{1}, \label{UAU=c1}%
\end{equation}
and define
\begin{equation}
M_{\theta}\left(  i\right)  :=\frac{1}{c\,}U_{i}M_{\theta}U_{i}^{\dagger},
\label{meas-rotate}%
\end{equation}
where
\begin{align*}
\left[  M_{+},M_{-}\right]   &  =0,\\
M_{\theta}  &  \geq0,\\
\mathrm{tr}\,M_{+}  &  =\mathrm{tr}\,M_{-}=1.
\end{align*}
Then, $\left\vert \Phi_{d}\right\rangle $ is universally optimal.

The proof is as follows. Observe
\begin{align*}
&  \left[  \Lambda_{+}\otimes\mathbf{I}\left(  \left\vert \Phi_{d}%
\right\rangle \left\langle \Phi_{d}\right\vert \right)  ,\Lambda_{-}%
\otimes\mathbf{I}\left(  \left\vert \Phi_{d}\right\rangle \left\langle
\Phi_{d}\right\vert \right)  \right] \\
&  =\frac{1}{d^{2}c\,^{2}}\sum_{i=1}^{m}\left\vert i\right\rangle \left\langle
i\right\vert \otimes\left[  \left(  U_{i}M_{+}U_{i}^{\dagger}\right)
^{T},\left(  U_{i}M_{-}U_{i}^{\dagger}\right)  ^{T}\right] \\
&  =0.
\end{align*}
Also,
\begin{align*}
&  \left\Vert \Lambda_{+}\otimes\mathbf{I}\left(  \left\vert \varphi
\right\rangle \left\langle \varphi\right\vert \right)  -s\,\Lambda_{-}%
\otimes\mathbf{I}\left(  \left\vert \varphi\right\rangle \left\langle
\varphi\right\vert \right)  \right\Vert _{1}\\
&  =\sum_{i=1}^{m}\left\Vert \sqrt{\rho_{\varphi}}\left(  M_{+}\left(
i\right)  ^{T}-s\,M_{-}\left(  i\right)  ^{T}\right)  \sqrt{\rho_{\varphi}%
}\right\Vert _{1}\\
&  \leq\sum_{i=1}^{m}\left\Vert \sqrt{\rho_{\varphi}}\left\vert M_{+}\left(
i\right)  ^{T}-s\,M_{-}\left(  i\right)  ^{T}\right\vert \sqrt{\rho_{\varphi}%
}\right\Vert _{1}\\
&  =\sum_{i=1}^{m}\mathrm{tr}\,\rho_{\varphi}\left\vert M_{+}\left(  i\right)
-s\,M_{-}\left(  i\right)  \right\vert ^{T}\\
&  =\frac{1}{c}\mathrm{tr}\,\rho_{\varphi}\left(  \sum_{i=1}^{m}%
U_{i}\left\vert M_{+}-s\,M_{-}\right\vert U_{i}^{\dagger}\right)  ^{T}\\
&  =\mathrm{tr}\,\rho_{\varphi}\mathrm{tr}\,\left\vert M_{+}-s\,M_{-}%
\right\vert =\mathrm{tr}\,\left\vert M_{+}-s\,M_{-}\right\vert ,
\end{align*}
where the inequality in the third line is true if $\left\vert \varphi
\right\rangle =\left\vert \Phi_{d}\right\rangle $. Therefore, by
Lemma\thinspace\ref{lem:commute}, we have the assertion.
\end{example}

\begin{example}
Let%
\begin{align*}
M_{\theta}\left(  2\right)   &  =\mathbf{1}-M_{\theta}\left(  1\right)
,\,\,\\
M_{-}\left(  2\right)   &  =M_{+}\left(  1\right)  =M=\mathrm{diag}\left(
a_{1},\cdots,a_{d}\right)  ,\\
a_{1}  &  >a_{2}\geq\cdots\geq a_{d}.
\end{align*}
Then,%
\begin{align*}
&  \left\Vert \Lambda_{+}\otimes\mathbf{I}\left(  \left\vert \psi\right\rangle
\left\langle \psi\right\vert \right)  -s\,\Lambda_{-}\otimes\mathbf{I}\left(
\left\vert \psi\right\rangle \left\langle \psi\right\vert \right)  \right\Vert
_{1}\\
&  =\left\Vert \sqrt{\rho_{\psi}}\left(  -s\mathbf{1}+\left(  1+s\right)
M^{T}\right)  \sqrt{\rho_{\psi}}\right\Vert _{1}+\left\Vert \sqrt{\rho_{\psi}%
}\left(  \mathbf{1}-\left(  1+s\right)  M^{T}\right)  \sqrt{\rho_{\psi}%
}\right\Vert _{1}\\
&  \leq\mathrm{tr}\,\rho_{\psi}\left(  \left\vert -s\mathbf{1}+\left(
1+s\right)  M^{T}\right\vert +\left\vert \mathbf{1}-\left(  1+s\right)
M^{T}\right\vert \right) \\
&  =\sum_{i=1}^{d}\rho_{\psi,i,i}\left(  \left\vert -s+\left(  1+s\right)
a_{i}\right\vert +\left\vert 1-\left(  1+s\right)  a_{i}\right\vert \right) \\
&  =\sum_{i=1}^{d}\rho_{\psi,i,i}\left(  \left\vert \left(  1-a_{i}\right)
s-a_{i}\right\vert +\left\vert a_{i}s-\left(  1-a_{i}\right)  \right\vert
\right) \\
&  \leq\left\vert \left(  1-a_{1}\right)  s-a_{1}\right\vert +\left\vert
a_{1}s-\left(  1-a_{1}\right)  \right\vert ,
\end{align*}
and the inequalities in the third and the fourth line are achieved by
$\rho_{\psi}=\left\vert 1\right\rangle \left\langle 1\right\vert $. Therefore,
by Lemma\thinspace\ref{lem:commute}, $\,\left\vert \psi\right\rangle
=\left\vert 1\right\rangle \left\vert 1\right\rangle $ is universally optimal.
\end{example}

\section{\textrm{SU}$(d)$ family}

\label{sec:su(d)}

\subsection{$d=2$ case}

\label{subsec:su2}

In this subsection, we consider the family $\left\{  \Upsilon_{U}%
\,;\,U\in\mathrm{SU}\left(  2\right)  \right\}  $. \cite{Fujiwara:2002} had
shown that $\left\vert \Phi_{2}\right\rangle $ is optimal for the mean square
error with the constraint (\ref{EML=0}). Also, as stated in Theorem\thinspace
\ref{th:unital-qubit}, Section\thinspace\ref{sec:d=2}, $\left\vert \Phi
_{2}\right\rangle $ is a universally optimal state.

Define for $\left\vert \psi\right\rangle \in\mathcal{H}_{in}\otimes
\mathcal{H}_{R}$,
\[
\mathcal{U}\left(  \psi\right)  :=\left\{  U;\mathrm{tr}\,\rho_{\psi
}U=0,\,U\in\mathrm{SU}\left(  d\right)  \right\}  .
\]
Below, we consider the test between the two hypotheses, $U=\mathbf{1}$ v.s.
$U\in\mathcal{U}\left(  \psi\right)  $. In other words, $\mathcal{D}=\left\{
0,1\right\}  $ and the loss function $l^{\psi}$ is such that
\begin{align*}
l_{U}^{\psi}\left(  0\right)   &  =\left\{
\begin{array}
[c]{cc}%
1, & \left(  U\in\mathcal{U}\left(  \psi\right)  \right) \\
0, & \left(  U=\mathbf{1}\right)
\end{array}
\right.  ,\\
l_{U}^{\psi}\left(  1\right)   &  =\left\{
\begin{array}
[c]{cc}%
0, & \left(  U\in\mathcal{U}\left(  \psi\right)  \right) \\
1, & \left(  U=\mathbf{1}\right)
\end{array}
\right.  ,\\
l_{U}^{\psi}\left(  1\right)   &  =l_{U}^{\psi}\left(  0\right)
=0\,,U\notin\mathcal{U}\left(  \psi\right)  \cup\left\{  \mathbf{1}\right\}
\end{align*}

\begin{theorem}
Consider the family $\left\{  \Upsilon_{U}\,;\,U\in\mathrm{SU}\left(
2\right)  \right\}  $. Then, $\left\vert \Phi_{2}\right\rangle $ is strictly
universally optimal.
\end{theorem}

\begin{proof}
Consider the loss function $l^{\Phi_{2}}$.Then, since $\left\vert \Phi
_{2}\right\rangle $ and $U\otimes\mathbf{1}\left\vert \Phi_{2}\right\rangle $
($U\in\mathcal{U}\left(  \Phi_{2}\right)  $) are orthogonal, for any prior
distribution $\pi$,
\[
\min_{M}R\left(  l^{\Phi_{2}},M,\pi,\left\vert \Phi_{2}\right\rangle \right)
=0.
\]
This is not the case if the input $\left\vert \psi\right\rangle $ is not
maximally entangled. Indeed, any \ $U\in\mathrm{SU}\left(  2\right)  $ satisfy
$\left\vert U_{11}\right\vert =\left\vert U_{22}\right\vert $. Without loss of
generality, suppose the Schmidt basis of $\left\vert \psi\right\rangle
=\sqrt{p_{1}}\left\vert 1\right\rangle \left\vert 1\right\rangle +\sqrt{p_{2}%
}\left\vert 2\right\rangle \left\vert 2\right\rangle $, where $p_{1}\neq
p_{2}$ . Then, the inner product between $\left\vert \psi\right\rangle $ and
$U\otimes\mathbf{1}\left\vert \psi\right\rangle $ equals%
\[
\left\langle \psi\right\vert U\otimes\mathbf{1}\left\vert \psi\right\rangle
=p_{1}U_{11}+p_{2}U_{22}.
\]
But this cannot equal to $0$ because of $p_{1}\neq p_{2}$ and $\left\vert
U_{11}\right\vert =\left\vert U_{22}\right\vert $. Therefore,
\[
\min_{M}R\left(  l^{\Phi_{2}},M,\pi,\left\vert \psi\right\rangle \right)
\neq0.
\]

\end{proof}

\subsection{Tests on $\mathrm{SU}\left(  d\right)  $ ($d\geq3$)}

\label{subsec:sud}

This subsection is devoted to the proof of the following theorem.

\begin{theorem}
\label{th:su(d)}Consider the channel family $\left\{  \Upsilon_{U}%
\,;\,U\in\mathrm{SU}\left(  d\right)  \right\}  $, where $d\geq3$. Then, any
$\left\vert \psi\right\rangle \in\mathcal{H}_{in}\otimes\mathcal{H}_{R}$ is
strictly admissible.
\end{theorem}

First, we introduce a series of propositions and lemmas.

\begin{proposition}
\label{prop:include}Consider the channel family $\left\{  \Upsilon_{U}%
;U\in\mathrm{SU}\left(  d\right)  \right\}  $. Then, if $\mathcal{U}\left(
\psi\right)  -\mathcal{U}\left(  \psi^{\prime}\right)  \neq\emptyset$,
\[
\min_{M}R\left(  l^{\psi},M,\pi,\left\vert \psi^{\prime}\right\rangle \right)
>\min_{M}R\left(  l^{\psi},M,\pi,\left\vert \psi\right\rangle \right)
,\,\exists\pi.
\]

\end{proposition}

\begin{proof}
By the definition of $\mathcal{U}\left(  \psi\right)  $, there is a projective
binary measurement $\left\{  M_{0},M_{1}\right\}  $ which distinguishes
$\left\vert \psi\right\rangle $ and $\left\{  U\otimes\mathbf{1}\left\vert
\psi\right\rangle ;U\in\mathcal{U}\left(  \psi\right)  \right\}  $ without
error. Therefore, by the definition of $l^{\psi}$, \
\[
\min_{M}R\left(  l^{\psi},M,\pi,\left\vert \psi\right\rangle \right)
=0,\,\forall\pi.
\]
If $U\in\mathcal{U}\left(  \psi\right)  -\mathcal{U}\left(  \psi^{\prime
}\right)  \neq\emptyset$, \ then $\left\vert \psi^{\prime}\right\rangle $ and
$U\otimes\mathbf{1}\left\vert \psi^{\prime}\right\rangle $ can not be
distinguished perfectly. Therefore, \
\[
\min_{M}R\left(  l^{\psi},M,\pi,\left\vert \psi^{\prime}\right\rangle \right)
>0,\,\exists\pi.
\]
Therefore, we have the assertion.
\end{proof}

\ 

For $x\in%
\mathbb{R}
^{d}$, define
\[
Ang\left(  x\right)  :=\left\{  \vec{\omega};\vec{\omega}\in%
\mathbb{C}
^{d-1},\left\vert \omega_{i}\right\vert =1,\sum_{i=1}^{d-1}x_{i}\omega
_{i}+x_{d}=0\right\}  .
\]
Also, if
\[
x_{i_{1}}\geq x_{i_{2}}\geq\cdots\geq x_{i_{d}},
\]%
\[
x_{j}^{\downarrow}:=x_{i_{j}},\,\,x^{\downarrow}:=\left(  x_{1}^{\downarrow
},x_{2}^{\downarrow},\cdots,x_{d}^{\downarrow}\right)  .
\]

The proof of the following lemma will be given in Appendix\thinspace
\ref{appendix:pf-lem}

\begin{lemma}
\label{lem:parallel}Suppose \ $d\geq3$, $x_{i},x_{i}^{\prime}>0$
($i=1,\cdots,d$),\thinspace\ $x_{1}^{\downarrow}>\sum_{i=2}^{d}x_{i}%
^{\downarrow}$, and $x_{1}^{\prime\downarrow}>\sum_{i=2}^{d}x_{i}^{^{\prime
}\downarrow}$. If $Ang\left(  x\right)  \subset Ang\left(  x^{\prime}\right)
$, then $x^{\prime}=sx$ for some $s\in%
\mathbb{R}
$.
\end{lemma}

\begin{lemma}
\label{lem:kernel}Let $\rho_{0},\rho_{1}\in\mathcal{B}\left(
\mathbb{C}
^{d}\right)  $, and $\left\{  \left\vert i\right\rangle \right\}  _{i=1}^{d}$
be an orthonormal basis in $%
\mathbb{C}
^{d}$. \ Suppose an Hermitian matrix $A$ satisfies
\[
\left\langle i\right\vert U^{\dagger}AU\left\vert i\right\rangle =0,\left(
i=0,\cdots,d\right)  \,,
\]
for any unitary $U\in U\left(  d\right)  $ such that
\begin{equation}
\,0<\left\langle i\right\vert U^{\dagger}\rho_{j}U\left\vert i\right\rangle
<\frac{1}{2}\,\,\left(  i=1,\cdots,d,\,j=0,1\right)  . \label{u-constrain}%
\end{equation}
Then we have%
\[
A=0.
\]

\end{lemma}

\begin{proof}
Let
\[
E_{ij}:=\left\vert i\right\rangle \left\langle j\right\vert +\left\vert
j\right\rangle \left\langle i\right\vert ,\,F_{ij}:=\sqrt{-1}\left(
\left\vert i\right\rangle \left\langle j\right\vert -\left\vert j\right\rangle
\left\langle i\right\vert \right)  .
\]
Then, there are real numbers $a_{i}$, $b_{ij}$, and $c_{ij}$ with \
\[
U^{\dagger}AU=\sum_{i}a_{i}\left\vert i\right\rangle \left\langle i\right\vert
+\sum_{i>j}\left(  b_{ij}E_{ij}+c_{ij}F_{ij}\right)  .
\]
Due to $\left\langle i\right\vert U^{\dagger}AU\left\vert i\right\rangle =0$,
$a_{i}=0$.

If $U\in\mathrm{U}\left(  d\right)  $ satisfies (\ref{u-constrain}), then any
member of neighborhood of $U$ satisfies (\ref{u-constrain}). Therefore,
\[
\left\langle i\right\vert \left[  U^{\dagger}AU,H\right]  \left\vert
i\right\rangle =0,\,\left(  i=0,\cdots,d\right)  \,,
\]
for any Hermitian matrix $H$. Also,
\begin{equation}
\left[  E_{ij},F_{ij}\right]  =-2\sqrt{-1}\left(  \left\vert i\right\rangle
\left\langle i\right\vert +\left\vert j\right\rangle \left\langle j\right\vert
\right)  ,\nonumber
\end{equation}
and, if $\delta_{ik}\delta_{jl}=\delta_{il}\delta_{jk}=0$,
\[
\left\langle m\right\vert \left[  E_{ij},F_{kl}\right]  \left\vert
m\right\rangle =0.\,\,
\]
Therefore,
\begin{align*}
\left\langle i\right\vert \left[  U^{\dagger}AU,E_{ij}\right]  \left\vert
i\right\rangle  &  =-2\sqrt{-1}c_{ij},\\
\left\langle i\right\vert \left[  U^{\dagger}AU,F_{ij}\right]  \left\vert
i\right\rangle  &  =2\sqrt{-1}b_{ij}.
\end{align*}
Hence, $b_{ij}=c_{ij}=0$. After all, $U^{\dagger}AU=0$, implying $A=0$.
\end{proof}

Now, we are in the position to present the proof of Theorem\thinspace
\ref{th:su(d)}.

\begin{proof}
\textbf{(Theorem\thinspace\ref{th:su(d)})} In view of Proposition\thinspace
\ref{prop:include}, it suffices to show that $\mathcal{U}\left(  \psi\right)
\subset\mathcal{U}\left(  \psi^{\prime}\right)  $ implies $\rho_{\psi^{\prime
}}=\rho_{\psi}$. \ 

Suppose $\mathcal{U}\left(  \psi\right)  \subset\mathcal{U}\left(
\psi^{\prime}\right)  $. Let us define, for $x\in%
\mathbb{R}
^{d},$
\[
\widetilde{Ang}\left(  x\right)  :=\left\{  \vec{\omega};\vec{\omega}\in%
\mathbb{C}
^{d},\left\vert \omega_{i}\right\vert =1,\sum_{i=1}^{d}\omega_{i}x_{i}%
=0,\prod_{i=1}^{d}\omega_{i}=1\right\}  ,
\]
and, for a $d\times d$ matrix $A$,
\[
\mathrm{Diag}\left(  A\right)  :=\left(  A_{11},A_{22},\cdots,A_{dd}\right)
.
\]
Then,
\[
\mathcal{U}\left(  \rho_{\psi}\right)  =\left\{  U\mathrm{diag}\left(
\vec{\omega}\right)  U^{\dagger};U\in U\left(  d\right)  ,\vec{\omega}%
\in\widetilde{Ang}\left(  \mathrm{Diag}\left(  U^{\dagger}\rho_{\psi}U\right)
\right)  \,\right\}  .
\]
Therefore, $\mathcal{U}\left(  \psi\right)  \subset\mathcal{U}\left(
\psi^{\prime}\right)  $ implies
\[
\widetilde{Ang}\left(  \mathrm{Diag}\left(  U^{\dagger}\rho_{\psi}U\right)
\right)  \,\subset\widetilde{Ang}\left(  \mathrm{Diag}\left(  U^{\dagger}%
\rho_{\psi^{\prime}}U\right)  \right)  ,\,
\]
or equivalently
\[
Ang\left(  \mathrm{Diag}\left(  U^{\dagger}\rho_{\psi}U\right)  \right)
\subset Ang\left(  \mathrm{Diag}\left(  U^{\dagger}\rho_{\psi^{\prime}%
}U\right)  \right)
\]
for all $U\in\mathrm{SU}\left(  d\right)  $.

Here we have recourse to Lemma\thinspace\ref{lem:parallel}, with
$x_{i}:=\left\langle i\right\vert U^{\dagger}\rho_{\psi}U\left\vert
i\right\rangle $ and $x_{i}^{\prime}:=\left\langle i\right\vert U^{\dagger
}\rho_{\psi^{\prime}}U\left\vert i\right\rangle $. Suppose $x_{1}^{\downarrow
}>\sum_{i=2}^{d}x_{i}^{\downarrow}$, and $x_{1}^{\prime\downarrow}>\sum
_{i=2}^{d}x_{i}^{^{\prime}\downarrow}$. Due to $\mathrm{tr}\,U^{\dagger}%
\rho_{\psi}U=\mathrm{tr}\,U^{\dagger}\rho_{\psi^{\prime}}U=1$, this is
equivalent to
\begin{align*}
0  &  <\left\langle i\right\vert U^{\dagger}\rho_{\psi}U\left\vert
i\right\rangle <\frac{1}{2},\\
0  &  <\left\langle i\right\vert U^{\dagger}\rho_{\psi^{\prime}}U\left\vert
i\right\rangle <\frac{1}{2},\,\,(i=1,\cdots,d).
\end{align*}
Therefore, by Lemma\thinspace\ref{lem:parallel},
\[
\mathrm{Diag}\left(  U^{\dagger}\rho_{\psi}U\right)  =\mathrm{Diag}\left(
U^{\dagger}\rho_{\psi^{\prime}}U\right)  ,\,\,
\]
which leads to
\[
\mathrm{Diag}\left(  U^{\dagger}\left(  \rho_{\psi}-\rho_{\psi^{\prime}%
}\right)  U\right)  =0.
\]
Therefore, by Lemma\thinspace\ref{lem:kernel}, \ $\rho_{\psi}=\rho
_{\psi^{\prime}}$ . Thus we have the assertion.
\end{proof}

\section{Universal enhancement by entanglement}

\label{sec:entanglement}

\subsection{Subfamily of $\left\{  \Lambda_{2,\theta}^{\mathrm{gp}}\right\}
$}

Observe, in Example\thinspace\ref{ex:diag}, a separable state
\[
\left\vert \psi\right\rangle :=\frac{1}{\sqrt{d}}\sum_{i=1}^{d}\left\vert
i\right\rangle _{in}\left\vert f\right\rangle _{R}\succeq^{c}\left\vert
\Phi_{d}\right\rangle ,
\]
where $\left\vert f\right\rangle _{R}\in\mathcal{H}_{R}$ is arbitrary, is
universally optimal. Combining with the fact that $\left\vert \Phi
_{d}\right\rangle $ is universally optimal (Theorem\thinspace
\ref{th:covariant-opt-input}), we have
\begin{equation}
\left\vert \psi\right\rangle \equiv^{c}\left\vert \Phi_{d}\right\rangle .
\label{equiv}%
\end{equation}

The proof is as follows. For unitary operators $U_{i}\in U\left(
\mathcal{H}_{R}\right)  $ ($i=1,\cdots,d$), let $\ U_{1}\oplus\cdots\oplus
U_{d}$ be the unitary operator acting on $\mathcal{H}_{in}\otimes
\mathcal{H}_{R}=\mathcal{H}_{out}\otimes\mathcal{H}_{R}$ such that
\[
\left(  U_{1}\oplus\cdots\oplus U_{d}\right)  \left\vert i\right\rangle
_{in}\left\vert j\right\rangle _{R}=\left\vert i\right\rangle _{in}%
U_{i}\left\vert j\right\rangle _{R}.
\]
Then, if $U_{i}\left\vert f\right\rangle =\left\vert i\right\rangle $,
\[
\left(  U_{1}\oplus\cdots\oplus U_{d}\right)  \left\vert \psi\right\rangle
=\left\vert \Phi_{d}\right\rangle .
\]
Observe
\[
\left(  \Lambda_{d,\theta}^{\mathrm{diag}}\otimes\mathbf{I}\right)
\circ \Upsilon_{U_{1}\oplus\cdots\oplus U_{d}}=\Upsilon_{U_{1}\oplus
\cdots\oplus U_{d}}\circ\left(  \Lambda_{d,\theta}^{\mathrm{diag}}%
\otimes\mathbf{I}\right)  .
\]
Therefore, \quad\
\begin{align*}
\Upsilon_{U_{1}\oplus\cdots\oplus U_{d}}\circ\left(  \Lambda_{d,\theta
}^{\mathrm{diag}}\otimes\mathbf{I}\right)  \left(  \left\vert \psi
\right\rangle \left\langle \psi\right\vert \right)   &  =\left(
\Lambda_{d,\theta}^{\mathrm{diag}}\otimes\mathbf{I}\right)  \circ
\Upsilon_{U_{1}\oplus\cdots\oplus U_{d}}\left(  \left\vert \psi\right\rangle
\left\langle \psi\right\vert \right) \\
&  =\left(  \Lambda_{d,\theta}^{\mathrm{diag}}\otimes\mathbf{I}\right)
\left(  \left\vert \Phi_{d}\right\rangle \left\langle \Phi_{d}\right\vert
\right)  .
\end{align*}
Therefore, by proposition\thinspace\ref{prop:randomize}, we have the assertion
(\ref{equiv}).

So in this case, entanglement between $\mathcal{H}_{in}$ and $\mathcal{H}_{R}%
$\ is not necessary. When an entangled state is strictly universally better
than any separable state? Below, the condition for $\left\vert \Phi
_{d}\right\rangle $ to be strictly universally better than any separable
states is studied. After investigating the subfamily of $\left\{
\Lambda_{2,\theta}^{\mathrm{gp}}\right\}  $ in this subsection, we move to
families of measurements (Examples\thinspace\ref{ex:ortho-measure} and
\ref{ex:ortho-measure}) in the next subsection.

\begin{theorem}
\label{th:e-not-needed}With $\xi_{\theta}\in%
\mathbb{R}
^{3}$, let $\Lambda_{\theta}=\Lambda_{2,\xi_{\theta}}^{\mathrm{gp}}$. Then,
there is a separable state $\left\vert \psi\right\rangle =\left\vert \psi
_{in}\right\rangle \left\vert \psi_{R}\right\rangle $ with $\left\vert
\psi_{in}\right\rangle \left\vert \psi_{R}\right\rangle \succeq^{c}\left\vert
\Phi_{2}\right\rangle $, or equivalently
\begin{equation}
\left\{  \Lambda_{\theta}\left(  \left\vert \psi_{in}\right\rangle
\left\langle \psi_{in}\right\vert \right)  \right\}  _{\theta\in\Theta}%
\succeq^{c}\left\{  \Lambda_{\theta}\otimes\mathbf{I}\left(  \left\vert
\Phi_{2}\right\rangle \left\langle \Phi_{2}\right\vert \right)  \right\}
_{\theta\in\Theta} \label{e-not-needed}%
\end{equation}
if and only if

\begin{description}
\item[(i)] $\left\{  \xi_{\theta}\right\}  _{\theta\in\Theta}$ is on a
straight line.

\item[(ii)] If there is at least a pair $\theta_{1}$, $\theta_{2}$ such that
$\xi_{+}:=\xi_{\theta_{1}}$ and \ $\xi_{-}:=\xi_{\theta_{2}}$ are distinct,
$\xi_{+}$ , $\xi_{-}$ satisfies%
\begin{equation}
\xi_{\mathcal{+}}^{1}\xi_{-}^{0}=\xi_{-}^{1}\xi_{+}^{0},\,\,\xi_{\mathcal{+}%
}^{2}\xi_{\mathcal{-}}^{3}=\xi_{\mathcal{-}}^{2}\xi_{\mathcal{+}}^{3},
\label{nes-suff-1}%
\end{equation}
or%
\begin{equation}
\xi_{\mathcal{+}}^{2}\xi_{-}^{0}=\xi_{\mathcal{-}}^{2}\xi_{+}^{0}%
,\,\,\xi_{\mathcal{+}}^{3}\xi_{-}^{1}=\xi_{\mathcal{-}}^{3}\xi_{\mathcal{+}%
}^{1}, \label{nes-suff-2}%
\end{equation}
or
\begin{equation}
\xi_{\mathcal{+}}^{3}\xi_{-}^{0}=\xi_{\mathcal{-}}^{3}\xi_{+}^{0}%
,\,\,\xi_{\mathcal{+}}^{1}\xi_{\mathcal{-}}^{2}=\xi_{-}^{1}\xi_{\mathcal{+}%
}^{2}, \label{nes-suff-3}%
\end{equation}
where $\xi_{\pm}^{0}:=1-\xi_{\pm}^{1}-\xi_{\pm}^{2}-\xi_{\pm}^{3}$ .
\end{description}
\end{theorem}

\ \ For example, $\left\{  \Lambda_{\frac{1}{2},\xi}^{\mathrm{damp}}\right\}
$ in Example\thinspace\ref{ex:g-damp} does not satisfy (i). Therefore,
$\left\vert \Phi_{2}\right\rangle $ is strictly universally better than any
separable states. On the other hand, $\left\{  \Lambda_{2,\xi}^{\mathrm{diag}%
}\right\}  $ (Example\ref{ex:diag} ) satisfies the hypothesis of the theorem.
Therefore, there is a separable state which is as good as $\left\vert \Phi
_{2}\right\rangle $.

\begin{proof}
We first study the case where $\Theta=\left\{  +,-\right\}  $ \ and
$\Lambda_{+}=\Lambda_{2,\xi_{+}}^{\mathrm{gp}}$ and $\Lambda_{-}%
=\Lambda_{2,\xi_{-}}^{\mathrm{gp}}$. and give necessary and sufficient
conditions for (\ref{e-not-needed}).

Suppose (\ref{e-not-needed}) holds. Then, since $\Lambda_{+}\otimes
\mathbf{I}\left(  \left\vert \Phi_{2}\right\rangle \left\langle \Phi
_{2}\right\vert \right)  $ and $\Lambda_{-}\otimes\mathbf{I}\left(  \left\vert
\Phi_{2}\right\rangle \left\langle \Phi_{2}\right\vert \right)  $ commutes, by
Lemma\thinspace\ref{lem:comm-sufficient}, $\Lambda_{\mathcal{+}}\left(
\left\vert \psi_{in}\right\rangle \left\langle \psi_{in}\right\vert \right)  $
and $\Lambda_{\mathcal{-}}\left(  \left\vert \psi_{in}\right\rangle
\left\langle \psi_{in}\right\vert \right)  $ has to commute. Let $\vec{r}$ and
$\vec{r}_{\theta}$ be a Bloch vector of $\left\vert \psi_{in}\right\rangle
\left\langle \psi_{in}\right\vert $ and $\Lambda_{\theta}\left(  \left\vert
\psi_{in}\right\rangle \left\langle \psi_{in}\right\vert \right)  $,
respectively. Then, this means that
\begin{equation}
\vec{r}_{\mathcal{-}}=\alpha\vec{r}_{\mathcal{+}}, \label{proportional-r}%
\end{equation}
for a real number $\alpha$. Also,
\[
\vec{r}_{\theta}=\mathrm{diag}\left(  a_{\theta}^{1},a_{\theta}^{2},a_{\theta
}^{3}\right)  \,\vec{r},
\]
where
\begin{align*}
a_{\theta}^{1}  &  :=1-2\xi_{\theta}^{2}-2\xi_{\theta}^{3},\\
a_{\theta}^{2}  &  :=1-2\xi_{\theta}^{1}-2\xi_{\theta}^{3},\\
a_{\theta}^{3}  &  :=1-2\xi_{\theta}^{1}-2\xi_{\theta}^{2}.
\end{align*}

Let us denote by $\rho\left(  \vec{r}\right)  $ the state with Bloch vector
$\vec{r}$. By simple calculations, we can verify
\begin{align*}
&  \left\Vert \rho\left(  \vec{r}_{\mathcal{+}}\right)  -s\,\rho\left(
\alpha\vec{r}_{\mathcal{+}}\right)  \right\Vert _{1}\\
&  =\left\vert \frac{1}{2}+\left\Vert \vec{r}_{\mathcal{+}}\right\Vert
-s\left(  \frac{1}{2}+\alpha\left\Vert \vec{r}_{\mathcal{+}}\right\Vert
\right)  \right\vert +\left\vert \frac{1}{2}-\left\Vert \vec{r}_{\mathcal{+}%
}\right\Vert -s\left(  \frac{1}{2}-\alpha\left\Vert \vec{r}_{\mathcal{+}%
}\right\Vert \right)  \right\vert
\end{align*}
is non-decreasing in $\left\Vert \vec{r}_{\mathcal{+}}\right\Vert $ for any
$s\geq0$. Therefore, by Lemma\thinspace\ref{lem:commute},
\[
\left\{  \rho\left(  \vec{r}_{\mathcal{+}}\right)  ,\rho\left(  \alpha\vec
{r}_{\mathcal{+}}\right)  \right\}  \succeq^{c}\left\{  \rho\left(  \vec
{r}_{\mathcal{+}}^{\prime}\right)  ,\rho\left(  \alpha\vec{r}_{\mathcal{+}%
}^{\prime}\right)  \right\}  ,
\]
if and only if $\left\Vert \vec{r}_{\mathcal{+}}\right\Vert \geq\left\Vert
\vec{r}_{\mathcal{+}}^{\prime}\right\Vert $.

Therefore, we concentrate on $\vec{r}$ which maximizes $\left\Vert \vec
{r}_{\mathcal{+}}\right\Vert =\left\Vert \mathrm{diag}\left(  a_{\mathcal{+}%
}^{1},a_{\mathcal{+}}^{2},a_{\mathcal{+}}^{3}\right)  \,\vec{r}\right\Vert $.
This maximum can be achieved at $\vec{r}=\left(  1,0,0\right)  $, $\left(
0,1,0\right)  $, or $\left(  0,0,1\right)  $. Therefore, (\ref{e-not-needed})
holds if and only if
\begin{equation}
\left\{  \rho\left(  a_{\mathcal{+}}^{1},0,0\right)  ,\rho\left(  a_{-}%
^{1},0,0\right)  \right\}  \succeq^{c}\left\{  \Lambda_{\theta}\otimes
\mathbf{I}\left(  \left\vert \Phi_{2}\right\rangle \left\langle \Phi
_{2}\right\vert \right)  \right\}  _{\theta\in\left\{  \mathcal{+}%
,\mathcal{-}\right\}  }, \label{e-not-needed-1}%
\end{equation}
or%
\begin{equation}
\left\{  \rho\left(  0,a_{\mathcal{+}}^{2},0\right)  ,\rho\left(
0,a_{\mathcal{-}}^{2},0\right)  \right\}  \succeq^{c}\left\{  \Lambda_{\theta
}\otimes\mathbf{I}\left(  \left\vert \Phi_{2}\right\rangle \left\langle
\Phi_{2}\right\vert \right)  \right\}  _{\theta\in\left\{  \mathcal{+}%
,\mathcal{-}\right\}  }, \label{e-not-needed-2}%
\end{equation}
or
\begin{equation}
\left\{  \rho\left(  0,0,a_{\mathcal{+}}^{3}\right)  ,\rho\left(
0,0,a_{\mathcal{-}}^{3}\right)  \right\}  \succeq^{c}\left\{  \Lambda_{\theta
}\otimes\mathbf{I}\left(  \left\vert \Phi_{2}\right\rangle \left\langle
\Phi_{2}\right\vert \right)  \right\}  _{\theta\in\left\{  \mathcal{+}%
,\mathcal{-}\right\}  }. \label{e-not-needed-3}%
\end{equation}

Suppose (\ref{e-not-needed-1}) is the case. Then, in view of Lemma\thinspace
\ref{lem:commute}, we have to have%
\begin{align}
&  \left\Vert \rho\left(  a_{\mathcal{+}}^{1},0,0\right)  -s\,\rho\left(
a_{-}^{1},0,0\right)  \right\Vert _{1}\nonumber\\
&  =\left\vert \xi_{+}^{0}+\xi_{+}^{1}-s\left(  \xi_{-}^{0}+\xi_{-}%
^{1}\right)  \right\vert +\left\vert \xi_{\mathcal{+}}^{2}+\xi_{\mathcal{+}%
}^{3}-s\left(  \xi_{\mathcal{-}}^{2}+\xi_{\mathcal{-}}^{3}\right)  \right\vert
\nonumber\\
&  \geq\left\vert \xi_{+}^{0}-s\xi_{+}^{0}\right\vert +\left\vert
\xi_{\mathcal{+}}^{1}-s\xi_{-}^{1}\right\vert +\left\vert \xi_{\mathcal{+}%
}^{2}-s\xi_{\mathcal{-}}^{2}\right\vert +\left\vert \xi_{\mathcal{+}}^{3}%
-s\xi_{\mathcal{-}}^{3}\right\vert \nonumber\\
&  =\left\Vert \Lambda_{\mathcal{+}}\otimes\mathbf{I}\left(  \left\vert
\Phi_{2}\right\rangle \left\langle \Phi_{2}\right\vert \right)  -s\,\Lambda
_{\mathcal{-}}\otimes\mathbf{I}\left(  \left\vert \Phi_{2}\right\rangle
\left\langle \Phi_{2}\right\vert \right)  \right\Vert _{1}\,, \label{rho-srho}%
\end{align}
where%
\[
\xi_{\pm}^{0}:=1-\xi_{\pm}^{1}-\xi_{\pm}^{2}-\xi_{\pm}^{3}.
\]
On the other hand, observe%
\begin{align*}
\left\vert \xi_{+}^{0}+\xi_{+}^{1}-s\left(  \xi_{-}^{0}+\xi_{-}^{1}\right)
\right\vert  &  \leq\left\vert \xi_{+}^{0}-s\xi_{-}^{0}\right\vert +\left\vert
\xi_{\mathcal{+}}^{1}-s\xi_{-}^{1}\right\vert ,\\
\left\vert \xi_{\mathcal{+}}^{2}+\xi_{\mathcal{+}}^{3}-s\left(  \xi
_{\mathcal{-}}^{2}+\xi_{\mathcal{-}}^{3}\right)  \right\vert  &
\leq\left\vert \xi_{\mathcal{+}}^{2}-s\xi_{\mathcal{-}}^{2}\right\vert
+\left\vert \xi_{\mathcal{+}}^{3}-s\xi_{\mathcal{-}}^{3}\right\vert .
\end{align*}
Therefore, the inequality (\ref{rho-s rho}) is true for any $s\geq0$ if and
only if identities in above two inequalities hold for any $s\geq0$. Therefore,
(\ref{e-not-needed-1}) if and only if (\ref{nes-suff-1}).

Similarly, (\ref{e-not-needed-2}) and (\ref{e-not-needed-3}) holds if and only
if (\ref{nes-suff-2}) and (\ref{nes-suff-3}), respectively. Therefore, in the
case of $\Theta=\left\{  +,-\right\}  $, there is $\left\vert \psi
_{in}\right\rangle $ with (\ref{e-not-needed}) if and only if one of
(\ref{nes-suff-1}), (\ref{nes-suff-2}) or (\ref{nes-suff-3}) holds.

Next, we treat the case where $\Theta$ is an arbitrary set, and $\Lambda
_{\theta}=\Lambda_{2,\xi_{\theta}}^{\mathrm{gp}}$. We suppose that there is at
least a pair $\theta_{1}$, $\theta_{2}$ such that $\xi_{+}:=\xi_{\theta_{1}}$
and \ $\xi_{-}:=\xi_{\theta_{2}}$ are distinct. In view of Lemma\thinspace
\ref{lem:comm-sufficient}, (\ref{e-not-needed}) holds only if $\Lambda
_{\theta}\left(  \left\vert \psi_{in}\right\rangle \left\langle \psi
_{in}\right\vert \right)  $ and $\Lambda_{\theta^{\prime}}\left(  \left\vert
\psi_{in}\right\rangle \left\langle \psi_{in}\right\vert \right)  $ commutes
for any $\theta$, $\theta^{\prime}$. Therefore, $\left\{  \Lambda_{\theta
}\left(  \left\vert \psi_{in}\right\rangle \left\langle \psi_{in}\right\vert
\right)  \right\}  _{\theta\in\Theta}$ is on a straight line passing through
origin. Denoting $\Lambda_{\theta_{1}}$ and $\Lambda_{\theta_{2}}$ by
$\Lambda_{+}$ and $\Lambda_{-}$, respectively, for any $\theta\in\Theta$,
there is $\lambda_{\theta}\in%
\mathbb{R}
$ such that
\begin{equation}
\Lambda_{\theta}\left(  \left\vert \psi_{in}\right\rangle \left\langle
\psi_{in}\right\vert \right)  =\lambda_{\theta}\Lambda_{+}\left(  \left\vert
\psi_{in}\right\rangle \left\langle \psi_{in}\right\vert \right)  +\left(
1-\lambda_{\theta}\right)  \Lambda_{-}\left(  \left\vert \psi_{in}%
\right\rangle \left\langle \psi_{in}\right\vert \right)  .
\label{lambda-theta}%
\end{equation}
We assert (\ref{e-not-needed}) holds if and only if
\begin{align}
&  \left\{  \Lambda_{+}\left(  \left\vert \psi_{in}\right\rangle \left\langle
\psi_{in}\right\vert \right)  ,\Lambda_{-}\left(  \left\vert \psi
_{in}\right\rangle \left\langle \psi_{in}\right\vert \right)  \right\}
_{\theta\in\Theta}\nonumber\\
&  \succeq^{c}\left\{  \Lambda_{+}\otimes\mathbf{I}\left(  \left\vert \Phi
_{2}\right\rangle \left\langle \Phi_{2}\right\vert \right)  ,\Lambda
_{-}\otimes\mathbf{I}\left(  \left\vert \Phi_{2}\right\rangle \left\langle
\Phi_{2}\right\vert \right)  \right\}  \label{e-not-needed-b}%
\end{align}
and
\begin{equation}
\Lambda_{\theta}=\lambda_{\theta}\Lambda_{+}+\left(  1-\lambda_{\theta
}\right)  \Lambda_{-}\,. \label{lambda-theta-2}%
\end{equation}
The statement of the present theorem follows immediately from this assertion.

First, we show `only if '. Obviously, (\ref{e-not-needed}) implies
(\ref{e-not-needed-b}). Also, due to (\ref{e-not-needed}), for any positive
operator $F\leq\mathbf{1}$ there is a positive operator $F^{\prime}%
\leq\mathbf{1}$ such that%
\begin{align*}
&  \mathrm{tr}\,\left[  \left\{  \lambda_{\theta}\Lambda_{+}+\left(
1-\lambda_{\theta}\right)  \Lambda_{-}-\Lambda_{\theta}\right\}
\otimes\mathbf{I}\left(  \left\vert \Phi_{2}\right\rangle \left\langle
\Phi_{2}\right\vert \right)  \right]  F\\
&  =\mathrm{tr}\,\left[  \left\{  \lambda_{\theta}\Lambda_{+}+\left(
1-\lambda_{\theta}\right)  \Lambda_{-}-\Lambda_{\theta}\right\}  \left(
\left\vert \psi_{in}\right\rangle \left\langle \psi_{in}\right\vert \right)
\right]  F^{\prime}\\
&  =0.
\end{align*}
Here the second identity is due to (\ref{lambda-theta}). Since $F\leq
\mathbf{1}$ is arbitrary, we have
\[
\Lambda_{\theta}\otimes\mathbf{I}\left(  \left\vert \Phi_{2}\right\rangle
\left\langle \Phi_{2}\right\vert \right)  =\left\{  \lambda_{\theta}%
\Lambda_{+}+\left(  1-\lambda_{\theta}\right)  \Lambda_{-}\right\}
\otimes\mathbf{I}\left(  \left\vert \Phi_{2}\right\rangle \left\langle
\Phi_{2}\right\vert \right)  .
\]
Therefore, $\left\{  \Lambda_{\theta}\otimes\mathbf{I}\left(  \left\vert
\Phi_{0}\right\rangle \left\langle \Phi_{0}\right\vert \right)  \right\}
_{\theta\in\Theta}$ is also on a straight line, and so is $\left\{
\Lambda_{\theta}\right\}  _{\theta\in\Theta}$. Thus we have
(\ref{lambda-theta-2}).

To show the opposite, suppose (\ref{e-not-needed-b}) and (\ref{lambda-theta-2}%
) holds. Then, for any measurement $M$, there exists a measurement $M^{\prime
}$ such that%
\begin{align*}
P_{\Lambda_{\theta}\otimes\mathbf{I}\left(  \left\vert \Phi_{2}\right\rangle
\left\langle \Phi_{2}\right\vert \right)  }^{M}  &  =\lambda_{\theta
}\,P_{\Lambda_{+}\otimes\mathbf{I}\left(  \left\vert \Phi_{2}\right\rangle
\left\langle \Phi_{2}\right\vert \right)  }^{M}+\left(  1-\lambda_{\theta
}\right)  \,P_{\Lambda_{-}\otimes\mathbf{I}\left(  \left\vert \Phi
_{2}\right\rangle \left\langle \Phi_{2}\right\vert \right)  }^{M}\\
&  =\lambda_{\theta}\,P_{\Lambda_{+}\left(  \left\vert \psi_{in}\right\rangle
\left\langle \psi_{in}\right\vert \right)  }^{M^{\prime}}+\left(
1-\lambda_{\theta}\right)  \,P_{\Lambda_{-}\left(  \left\vert \psi
_{in}\right\rangle \left\langle \psi_{in}\right\vert \right)  }^{M^{\prime}}\\
&  =\,P_{\Lambda_{\theta}\left(  \left\vert \psi_{in}\right\rangle
\left\langle \psi_{in}\right\vert \right)  }^{M^{\prime}}.
\end{align*}
Hence, we have (\ref{e-not-needed}), and our assertion is proved. Thus, we
have Theorem\thinspace\ref{th:e-not-needed}.
\end{proof}

\bigskip

\subsection{A pair of measurements}

In this subsection, we investigate the measurement families studied in
Examples\thinspace\ref{ex:ortho-measure},\thinspace\ref{ex:mes-rotate} of
Section\thinspace\ref{sec:measurement}.

First, in Example\thinspace\ref{ex:ortho-measure}, $\left\vert \psi
_{in}\right\rangle \left\vert \psi_{R}\right\rangle \succeq^{c}\left\vert
\Phi_{d}\right\rangle $ holds if and only if its output can be discriminated
with certainty, or equivalently,%
\begin{equation}
\left\langle \psi_{in}\right\vert M_{+}\left(  i\right)  \left\vert \psi
_{in}\right\rangle \left\langle \psi_{in}\right\vert M_{-}\left(  i\right)
\left\vert \psi_{in}\right\rangle =0,\,i=1,\cdots,m. \label{meas-conclusive}%
\end{equation}

\begin{proposition}
\label{prop:meas-ortho} In case of (\ref{meas-ortho-1}),
(\ref{meas-conclusive}) holds if and only if either
\begin{equation}
M_{+}\left(  i\right)  \left\vert \psi_{in}\right\rangle =0 \label{MM=0}%
\end{equation}
or
\begin{equation}
M_{+}\left(  i\right)  =c\,\left\vert \psi_{in}\right\rangle \left\langle
\psi_{in}\right\vert \,\,\,\,(c:\text{constant}) \label{M-prop-M}%
\end{equation}
holds for any $i$.
\end{proposition}

\begin{proof}
If $\left\langle \psi_{in}\right\vert M_{+}\left(  j\right)  \left\vert
\psi_{in}\right\rangle =0$, we have (\ref{MM=0}). On the other hand, suppose
$\left\langle \psi_{in}\right\vert M_{+}\left(  i\right)  \left\vert \psi
_{in}\right\rangle \neq0$. Then, for (\ref{meas-conclusive}) to be true,
$\left\langle \psi_{in}\right\vert M_{-}\left(  i\right)  \left\vert \psi
_{in}\right\rangle =0$ has to hold. Therefore, by (\ref{meas-ortho-1}),
\[
\mathrm{tr}\,M_{+}\left(  i\right)  =\left\langle \psi_{in}\right\vert
M_{+}\left(  i\right)  \left\vert \psi_{in}\right\rangle .
\]
Since $M_{+}\left(  i\right)  $'s rank is one, this holds if and only if
(\ref{M-prop-M}).
\end{proof}

Finally, we investigate Example \ref{ex:mes-rotate}. \thinspace Let
\[
M_{+}=\sum_{i=1}^{d}\alpha_{i}\left\vert e_{i}\right\rangle \left\langle
e_{i}\right\vert ,\,M_{+}=\sum_{i=1}^{d}\beta_{i}\left\vert e_{i}\right\rangle
\left\langle e_{i}\right\vert ,
\]
where $\left\{  \left\vert e_{i}\right\rangle \right\}  _{i=1}^{d}$ is an
orthonormal basis of $\mathcal{H}_{in}$. Below, we assume
\begin{equation}
\frac{\alpha_{1}}{\beta_{1}}>\frac{\alpha_{2}}{\beta_{2}}>\cdots>\frac
{\alpha_{d}}{\beta_{d}}. \label{a/b}%
\end{equation}

The proof of the following lemma is in Appendix\thinspace\ref{appendix:a/b}.

\begin{lemma}
\label{lem:a/b}Suppose $\alpha_{i}\geq0$, $\beta_{i}>0$, and (\ref{a/b})
holds. Suppose also $\sum_{i=1}^{m}\left\vert \gamma_{i}\right\vert ^{2}>0$.
Then,
\[
\frac{\alpha_{1}}{\beta_{1}}=\frac{\sum_{j=1}^{d}\left\vert \gamma
_{j}\right\vert ^{2}\alpha_{j}}{\sum_{j=1}^{d}\left\vert \gamma_{j}\right\vert
^{2}\beta_{j}}%
\]
holds if and only if $\left\vert \gamma_{1}\right\vert \neq0$ and
\[
\gamma_{2}=\gamma_{3}=\cdots=\gamma_{d}=0.
\]

\end{lemma}

\begin{lemma}
\label{lem:sumUBU=1}Suppose unitary matrices $\left\{  U_{i}\right\}
_{i=1}^{m}$ satisfies
\begin{equation}
\sum_{i=1}^{m}U_{i}AU_{i}^{\dagger}=\left(  c\mathrm{tr}\,A\right)
\mathbf{1.} \label{UAU=c1-2}%
\end{equation}
Then,
\[
\sum_{i=1}^{m}U_{i}^{\dagger}BU_{i}=\left(  c\mathrm{tr}\,B\right)
\mathbf{1.}%
\]

\end{lemma}

\begin{proof}
By (\ref{UAU=c1-2}), we have%
\begin{align*}
c\,\mathrm{tr}\,B\mathrm{tr}\,A  &  =\sum_{i=1}^{m}\mathrm{tr}\,\,BU_{i}%
AU_{i}^{\dagger}\,\\
&  =\sum_{i=1}^{m}\mathrm{tr}\,\,U_{i}^{\dagger}BU_{i}A.
\end{align*}
Since this holds for any $A$, we have the assertion.
\end{proof}

\begin{proposition}
\label{prop:mes-rotate}In Example \ref{ex:mes-rotate}, suppose $\alpha_{i}%
\geq0$, $\beta_{i}>0$. Also suppose (\ref{a/b}) holds. Then, $\left\vert
\psi_{in}\right\rangle \left\vert \psi_{R}\right\rangle \succeq^{c}\left\vert
\Phi_{d}\right\rangle $ is equivalent to the following: there is a surjection
$f:\left\{  1,\cdots,m\right\}  \rightarrow\left\{  1,\cdots,d\right\}  $, a
state vector $\left\vert \varphi\right\rangle $ and unimodular complex numbers
$\omega_{i}$ ($i=1$,$\cdots$,$d$) such that
\begin{align*}
\left\vert e_{f\left(  i\right)  }\right\rangle  &  =\omega_{i}U_{i}^{\dagger
}\left\vert \psi_{in}\right\rangle \,,\,\,(i=1,2,\cdots,m),\\
c  &  =\left\vert \left\{  i\,\,;\,\,f\left(  i\right)  =j\right\}
\right\vert \,\,(j=1,2,\cdots,d).
\end{align*}

\end{proposition}

\begin{proof}
For $\left\vert \psi_{in}\right\rangle \left\vert \psi_{R}\right\rangle
\succeq^{c}\left\vert \Phi_{d}\right\rangle $ to hold, we have to have
\begin{align*}
\mathrm{tr}\,\left\vert M_{+}-sM_{-}\right\vert  &  =\sum_{i=1}^{d}\left\vert
\alpha_{i}-s\beta_{i}\right\vert \\
&  =\frac{1}{c}\sum_{i=1}^{m}\,\left\vert \left\langle \psi_{in}\right\vert
U_{i}\left(  M_{+}-sM_{-}\right)  U_{i}^{\dagger}\left\vert \psi
_{in}\right\rangle \right\vert ,\,\forall s\geq0\text{.}%
\end{align*}
Therefore, by defining $f$ properly, we have to have%
\begin{align*}
\alpha_{j}  &  =\frac{1}{c}\sum_{i:\,f\left(  i\right)  =j}\,\left\langle
\psi_{in}\right\vert U_{i}M_{+}U_{i}^{\dagger}\left\vert \psi_{in}%
\right\rangle ,\\
\beta_{j}  &  =\,\frac{1}{c}\sum_{i:\,f\left(  i\right)  =j}\,\left\langle
\psi_{in}\right\vert U_{i}M_{-}U_{i}^{\dagger}\left\vert \psi_{in}%
\right\rangle ,\,\,\left(  j=1,\cdots,d\right)  .
\end{align*}
Thus, if $f\left(  i\right)  =1$,
\begin{align*}
\,\frac{\alpha_{1}}{\beta_{1}}  &  =\frac{\left\langle \psi_{in}\right\vert
U_{i}M_{+}U_{i}^{\dagger}\left\vert \psi_{in}\right\rangle }{\left\langle
\psi_{in}\right\vert U_{i}M_{-}U_{i}^{\dagger}\left\vert \psi_{in}%
\right\rangle }\\
&  =\frac{\sum_{j=1}^{d}\left\vert \gamma_{i,j}\right\vert ^{2}\alpha_{j}%
}{\sum_{j=1}^{d}\left\vert \gamma_{i,j}\right\vert ^{2}\beta_{j}},
\end{align*}
where
\[
U_{i}^{\dagger}\left\vert \psi_{in}\right\rangle =\sum_{j=1}^{d}\gamma
_{i,j}\left\vert e_{j}\right\rangle .
\]
By Lemma\thinspace\ref{lem:a/b}, then we should have
\[
\omega_{i}U_{i}^{\dagger}\left\vert \psi\right\rangle =\left\vert
e_{1}\right\rangle .
\]
Thus,%
\[
\alpha_{1}=\frac{1}{c}\sum_{i:\,f\left(  i\right)  =1}\,\left\langle \psi
_{in}\right\vert U_{i}M_{+}U_{i}^{\dagger}\left\vert \psi_{in}\right\rangle
=\frac{1}{c}\left\vert \left\{  i\,;\,f\left(  i\right)  =1\right\}
\right\vert \alpha_{1},
\]
or
\begin{equation}
c=\left\vert \left\{  i\,;\,f\left(  i\right)  =1\right\}  \right\vert .
\label{c=||}%
\end{equation}
Therefore,
\begin{equation}
\frac{1}{c}\sum_{i\,:\,f\left(  i\right)  =1}U_{i}^{\dagger}\left\vert
\psi_{in}\right\rangle \left\langle \psi_{in}\right\vert U_{i}=\left\vert
e_{1}\right\rangle \left\langle e_{1}\right\vert . \label{sum-psi=e}%
\end{equation}
Since by Lemma\thinspace\ref{lem:sumUBU=1}
\[
\frac{1}{c}\sum_{i\,=1}^{m}U_{i}^{\dagger}\left\vert \psi_{in}\right\rangle
\left\langle \psi_{in}\right\vert U_{i}=\mathbf{1}%
\]
holds, we should have
\begin{equation}
\frac{1}{c}\sum_{i\,:\,f\left(  i\right)  \neq1}U_{i}^{\dagger}\left\vert
\psi_{in}\right\rangle \left\langle \psi_{in}\right\vert U_{i}=\sum_{j=2}%
^{d}\left\vert e_{j}\right\rangle \left\langle e_{j}\right\vert .
\label{sum-psi-e2}%
\end{equation}

Therefore, with $f\left(  i\right)  >1$, we should have
\[
\omega_{i}U_{i}^{\dagger}\left\vert \psi_{in}\right\rangle =\sum_{j=2}%
^{d}\gamma_{j}\left\vert e_{j}\right\rangle ,
\]
and if $f\left(  i\right)  =2$,
\begin{align*}
\,\frac{\alpha_{2}}{\beta_{2}}  &  =\frac{\left\vert \left\langle \psi
_{in}\right\vert U_{i}M_{+}U_{i}^{\dagger}\left\vert \psi_{in}\right\rangle
\right\vert }{\left\vert \left\langle \psi_{in}\right\vert U_{i}M_{-}%
U_{i}^{\dagger}\left\vert \psi_{in}\right\rangle \right\vert }\\
&  =\frac{\sum_{j=2}^{d}\left\vert \gamma_{j}\right\vert ^{2}\alpha_{j}}%
{\sum_{j=2}^{d}\left\vert \gamma_{j}\right\vert ^{2}\beta_{j}}.
\end{align*}
Then, by Lemma\thinspace\ref{lem:a/b}, we should have
\[
\omega_{i}U_{i}^{\dagger}\left\vert \psi_{in}\right\rangle =\left\vert
e_{2}\right\rangle .
\]
Therefore, using the same argument as the one derived (\ref{c=||}),
(\ref{sum-psi=e}), and (\ref{sum-psi-e2}), we have
\begin{align*}
c  &  =\left\vert \left\{  i\,;\,f\left(  i\right)  =2\right\}  \right\vert
,\\
\frac{1}{c}\sum_{i\,:\,f\left(  i\right)  =2}U_{i}^{\dagger}\left\vert
\psi_{in}\right\rangle \left\langle \psi_{in}\right\vert U_{i}  &  =\left\vert
e_{2}\right\rangle \left\langle e_{2}\right\vert ,\\
\frac{1}{c}\sum_{i\,:\,f\left(  i\right)  \neq1,2}U_{i}^{\dagger}\left\vert
\psi_{in}\right\rangle \left\langle \psi_{in}\right\vert U_{i}  &  =\sum
_{i=3}^{d}\left\vert e_{i}\right\rangle \left\langle ei\right\vert .
\end{align*}
Recursively, for each $j$, we obtain
\begin{align*}
\omega_{i}U_{i}^{\dagger}\left\vert \psi_{in}\right\rangle  &  =\left\vert
e_{\tilde{f}\left(  i\right)  }\right\rangle \\
c  &  =\left\vert \left\{  i\,;\,f\left(  i\right)  =j\right\}  \right\vert \\
\frac{1}{c}\sum_{i\,:\,f\left(  i\right)  =j}U_{i}^{\dagger}\left\vert
\psi_{in}\right\rangle \left\langle \psi_{in}\right\vert U_{i}  &  =\left\vert
e_{j}\right\rangle \left\langle e_{j}\right\vert ,\\
\frac{1}{c}\sum_{i\,:\,f\left(  i\right)  \neq1,2,\cdots,j}U_{i}^{\dagger
}\left\vert \psi_{in}\right\rangle \left\langle \psi_{in}\right\vert U_{i}  &
=\sum_{i=j+1}^{d}\left\vert e_{i}\right\rangle \left\langle e_{i}\right\vert .
\end{align*}
Thus we obtain the assertion of the proposition.
\end{proof}

\subsection{Entanglement breaking channels which requires entanglement}

In \cite{Sacchi:2005:3}, they had shown that Bayes error probability of
hypothesis testing of a pair of entanglement breaking channel is smaller with
a maximally entangled input state than with any separable input states.
Likewise, we point out that a maximally entangled state is universally optimal
and strictly universally better than any separable state for some families of
entanglement breaking channels. Such families of entanglement breaking
channels can be composed using Theorem\thinspace\ref{th:e-not-needed} and
Propositions\thinspace\ref{prop:meas-ortho} and \ref{prop:mes-rotate}.

First, let us compose such family in the form of $\left\{  \Lambda
_{2,\xi_{\theta}}^{\mathrm{gp}}\right\}  _{\not \theta \in\Theta}$ using
Theorem\thinspace\ref{th:e-not-needed}. Observe $\Lambda_{2,\xi_{\theta}%
}^{\mathrm{gp}}$ is entanglement breaking if and only if \ $\Lambda
_{2,\xi_{\theta}}^{\mathrm{gp}}\otimes\mathbf{I}\left(  \left\vert \Phi
_{2}\right\rangle \left\langle \Phi_{2}\right\vert \right)  $ is separable. By
PPT criteria\thinspace\cite{Horodecki:1996}, this is equivalent to
\begin{equation}
\xi_{\theta}^{0}+\xi_{\theta}^{3}\geq\left\vert \xi_{\theta}^{1}-\xi_{\theta
}^{2}\right\vert ,\,\xi_{\theta}^{1}+\xi_{\theta}^{2}\geq\left\vert
\xi_{\theta}^{0}-\xi_{\theta}^{3}\right\vert , \label{gdep-e-break}%
\end{equation}
where $\xi_{\theta}^{0}:=1-\xi_{\theta}^{1}-\xi_{\theta}^{2}-\xi_{\theta}^{3}%
$. If a family$\left\{  \xi_{\theta}\right\}  _{\theta\in\Theta}$ satisfies
does not satisfy the hypothesis of Theorem\thinspace\ref{th:e-not-needed} is
not true, $\left\{  \Lambda_{2,\xi_{\theta}}^{\mathrm{gp}}\right\}
_{\not \theta \in\Theta}$ is an example of a family of channels with desired
properties. In particular, if $\left\{  \xi_{\theta}\right\}  _{\theta
\in\Theta}$ is not on the straight line, this is the case. Even if $\left\{
\xi_{\theta}\right\}  _{\theta\in\Theta}$ is on a straight line with
$\xi_{\theta_{1}}\neq\xi_{\theta_{2}}$, if no pair out of $\xi_{\theta_{1}%
}^{0}/\xi_{\theta_{2}}^{0}$, $\xi_{\theta_{1}}^{1}/\xi_{\theta_{2}}^{1}$,
$\xi_{\theta_{1}}^{2}/\xi_{\theta_{2}}^{2}$ and $\xi_{\theta_{1}}^{3}%
/\xi_{\theta_{2}}^{3}$ equals with each other, we also obtain an example of an
entanglement breaking channel with desired properties.

Second, consider POVM $\left\{  M_{+}\left(  i\right)  \right\}  $ \ such that
constituent operators are of unit rank and not orthogonal with each other.
Also, define POVM $\left\{  M_{-}\left(  i\right)  \right\}  $ by
(\ref{meas-ortho-1}). Then, by Proposition \ref{prop:meas-ortho}, the channel
family $\left\{  \Lambda_{\theta}\right\}  _{\theta\in\left\{  +,-\right\}  }$
defined via (\ref{measurement}) has desired property. For example, consider a
measurement with POVM
\begin{align*}
M_{+}\left(  1\right)   &  =\frac{1}{2a^{2}}\left[
\begin{array}
[c]{cc}%
a^{2} & ab\\
ab & b^{2}%
\end{array}
\right]  ,M_{+}\left(  2\right)  =\frac{1}{2a^{2}}\left[
\begin{array}
[c]{cc}%
a^{2} & -ab\\
-ab & b^{2}%
\end{array}
\right] \\
M_{+}\left(  3\right)   &  =\frac{1}{2a^{2}}\left[
\begin{array}
[c]{cc}%
0 & 0\\
0 & 2a^{2}-2b^{2}%
\end{array}
\right]  ,\\
M_{-}\left(  i\right)   &  =\mathrm{tr}\,M_{+}\left(  i\right)  \mathbf{1}%
-M_{+}\left(  i\right)  ,
\end{align*}
where $a>b>0$.

Finally, by Proposition\thinspace\ref{prop:mes-rotate}, we can add another set
of examples. \ Observe
\[
\sum_{i,j=0}^{d-1}X_{d}^{i}Z_{d}^{j}\,A\left(  X_{d}^{i}Z_{d}^{j}\right)
^{\dagger}=\mathbf{1},
\]
where $X_{d}$ , $Z_{d}$ are defined by (\ref{g-pauli}). By
Proposition\thinspace\ref{prop:mes-rotate}, if there is $\left\vert \psi
_{in}\right\rangle \left\vert \psi_{R}\right\rangle $ with $\left\vert
\psi_{in}\right\rangle \left\vert \psi_{R}\right\rangle \succeq^{c}\left\vert
\Phi_{d}\right\rangle $, we should have
\[
\left\vert e_{\tilde{f}\left(  i,j\right)  }\right\rangle =\omega_{ij}%
^{\prime}\left(  X_{d}^{i}Z_{d}^{j}\,\right)  ^{\dagger}\left\vert \psi
_{in}\right\rangle ,
\]
where $\tilde{f}\left(  i,j\right)  $ is a surjection onto $\left\{
1,\cdots,d\right\}  $ and $\omega_{ij}^{\prime}$ is a unimodular complex
number. Therefore, with $f\left(  i^{\prime},j^{\prime}\right)  =1$,%

\[
\left\vert e_{\tilde{f}\left(  i,j\right)  }\right\rangle =\omega_{ij}%
^{\prime}\overline{\omega_{i^{\prime}j^{\prime}}^{\prime}}\left(  X_{d}%
^{i}Z_{d}^{j}\,\right)  ^{\dagger}X_{d}^{i^{\prime}}Z_{d}^{j^{\prime}%
}\,\left\vert e_{1}\right\rangle .
\]
Therefore, by (\ref{g-pauli-property}), there is a surjection $f\left(
i,j\right)  $ onto $\left\{  1,\cdots,d\right\}  $ and a unimodular complex
number $\omega_{ij}$ with
\[
\left\vert e_{f\left(  i,j\right)  }\right\rangle =\omega_{ij}X_{d}^{i}%
Z_{d}^{j}\,\left\vert e_{1}\right\rangle .
\]

For example, let
\[
\left\vert e_{1}\right\rangle =\sum_{i=1}^{d}a_{i}\left\vert i\right\rangle ,
\]
where $a_{i}>0$ and $a_{i}\neq a_{j}$ ($i,j$). Then, $X_{d}\left\vert
e_{1}\right\rangle $ is neither parallel or orthogonal to $\left\vert
e_{1}\right\rangle $. Therefore, the conditions indicated by
Proposition\thinspace\ref{prop:mes-rotate} are not satisfied, and we have a
channel family with desired property.

\bigskip

\section{Iterative use of a channel}

\label{sec:iteration}

Allowed to use given channel $\Lambda_{\theta}$ for $n$ times, one may send in
identical $n$-copies of an input (\textit{identical repetition}), or create a
large entangled state in $\mathcal{H}_{in}^{\otimes n}$ and send to the
channels $\Lambda_{\theta}^{\otimes n}$ (\textit{parallel repetition}), or
modify the input depending on the output of the previous use of the channel
(\textit{sequential repetition}). By definition, an identical repetition is a
special case of a parallel repetition, which, in turn, is a special case of a
sequential repetition.

The final output state of the identical repetition with the input state
$\left\vert \psi\right\rangle ^{\otimes n}\in$ $\left(  \mathcal{H}%
_{in}\otimes\mathcal{H}_{R}\right)  ^{\otimes n}$ and the parallel repetition
with the input state $\left\vert \psi^{n}\right\rangle \in$ $\left(
\mathcal{H}_{in}\otimes\mathcal{H}_{R}\right)  ^{\otimes n}$ is%

\[
\rho_{\mathrm{if},\theta}^{n}=\left\{  \Lambda_{\theta}\otimes\mathbf{I}%
\left(  \left\vert \psi\right\rangle \left\langle \psi\right\vert \right)
\right\}  ^{\otimes n}\in\mathcal{B}\left(  \left(  \mathcal{H}_{out}%
\otimes\mathcal{H}_{R}\right)  ^{\otimes n}\right)  ,
\]
and%
\[
\rho_{\mathrm{pf},\theta}^{n}=\Lambda_{\theta}^{\otimes n}\otimes
\mathbf{I}\left(  \left\vert \psi^{n}\right\rangle \left\langle \psi
^{n}\right\vert \right)  \in\mathcal{B}\left(  \left(  \mathcal{H}%
_{out}\otimes\mathcal{H}_{R}\right)  ^{\otimes n}\right)  ,
\]
respectively. To describe the final output state of sequential repetition, we
introduce a series of Hilbert spaces $\left\{  \mathcal{H}_{in,i}\right\}
_{i=1}^{n}$, $\left\{  \mathcal{H}_{out,\,i}\right\}  _{i=1}^{n}$,
$\mathcal{H}_{R}^{n}$, where $\mathcal{H}_{in,\,i}\simeq\mathcal{H}_{in}$ and
$\mathcal{H}_{out,\,i}\simeq\mathcal{H}_{out}$ ($i=1,\cdots,n$), and a series
of completely positive trace preserving maps $\left\{  \Upsilon_{i}\right\}
_{i=1}^{n-1}$ from $\mathcal{B}\left(  \mathcal{H}_{out,i}\otimes
\mathcal{H}_{R}^{n}\right)  $ to $\mathcal{B}\left(  \mathcal{H}%
_{in,i+1}\otimes\mathcal{H}_{R}^{n}\right)  $. Here, dimension of
$\mathcal{H}_{R}^{n}$ is finite and large enough (in fact, $\dim
\mathcal{H}_{R}^{n}=\left(  \dim\mathcal{H}_{in}\right)  ^{n+1}\left(
\dim\mathcal{H}_{out}\right)  ^{n}$ \ is enough.) With the initial state
$\left\vert \psi\right\rangle \in\mathcal{H}_{in,\,1}\otimes\mathcal{H}%
_{R}^{n}$, the final output state of the sequential scheme is
\begin{align*}
\rho_{\mathrm{sf},\theta}^{n}  &  :=\left(  \Lambda_{\theta}\otimes
\mathbf{I}\right)  \circ \Upsilon_{n-1}\cdots\circ\left(  \Lambda_{\theta
}\otimes\mathbf{I}\right)  \circ \Upsilon_{2}\circ\left(  \Lambda_{\theta
}\otimes\mathbf{I}\right)  \circ \Upsilon_{1}\circ\left(  \Lambda_{\theta
}\otimes\mathbf{I}\right)  \left(  \left\vert \psi\right\rangle \left\langle
\psi\right\vert \right)  ,\\
&  \in\mathcal{B}\left(  \mathcal{H}_{out,\,n}\otimes\mathcal{H}_{R}%
^{n}\right)
\end{align*}
to which the measurement $M_{n}$ is applied.

\begin{theorem}
Let $\left\{  \Lambda_{\theta}\right\}  _{\theta\in\Theta}$ be covariant or
contravariant channels. Then, the universally optimal identical repetition
strategy achieves the figure of merit that can be achieved by the universally
optimal sequential repetition strategy. Here the optimal input state is
$\left\vert \psi_{opt}\right\rangle ^{\otimes n}$, where $\left\vert
\psi_{opt}\right\rangle $ is as of (\ref{covariant-opt}).
\end{theorem}

\begin{proof}
By Proposition\thinspace\ref{prop:randomize}, we only have to compose a CPTP
map $\tilde{\Gamma}^{n}$ with
\[
\rho_{\mathrm{sf},\theta}^{n}=\tilde{\Gamma}^{n}\left(  \left\{
\Lambda_{\theta}\otimes\mathbf{I}\left(  \left\vert \psi_{opt}\right\rangle
\left\langle \psi_{opt}\right\vert \right)  \right\}  ^{\otimes n}\right)  ,
\]
where, with $\mathcal{H}_{R}\simeq\mathcal{H}_{in}$ ,
\[
\left\{  \Lambda_{\theta}\otimes\mathbf{I}\left(  \left\vert \psi
_{opt}\right\rangle \left\langle \psi_{opt}\right\vert \right)  \right\}
^{\otimes n}\in\mathcal{B}\left(  \left(  \mathcal{H}_{out}\otimes
\mathcal{H}_{R}\right)  ^{\otimes n}\right)  .
\]

The composition of $\tilde{\Gamma}^{n}$ is as follows. Define $\mathcal{H}%
_{in,i}^{\prime}$, $\mathcal{H}_{R}^{\prime n}$ with the same dimension as
$\mathcal{H}_{in,i}$, $\mathcal{H}_{R}^{n}$ which would have used in the
sequential repetition protocol resulting the final state $\rho_{\mathrm{sf}%
,\theta}^{n}$. Prepare $\left\vert \psi\right\rangle $ in $\mathcal{H}%
_{in,1}^{\prime}\otimes\mathcal{H}_{R}^{\prime n}$, and apply $\Gamma$, which
is composed in the proof of Theorem\thinspace\ref{th:covariant-opt-input},
jointly to $\mathcal{H}_{in,1}^{\prime}$-part of $\left\vert \psi\right\rangle
$ and $\mathcal{H}_{R}$-part of $\Lambda_{\theta}\otimes\mathbf{I}\left(
\left\vert \psi_{opt}\right\rangle \left\langle \psi_{opt}\right\vert \right)
\in\mathcal{B}\left(  \mathcal{H}_{out}\otimes\mathcal{H}_{R}\right)  $,
producing $\left(  \Lambda_{\theta}\otimes\mathbf{I}\right)  \left(
\left\vert \psi\right\rangle \left\langle \psi\right\vert \right)  $ in the
space $\mathcal{H}_{out}\otimes\mathcal{H}_{R}^{n\prime}$. Then apply
$\Upsilon_{1}$, producing $\Upsilon_{1}\circ\left(  \Lambda_{\theta}%
\otimes\mathbf{I}\right)  \left(  \left\vert \psi\right\rangle \left\langle
\psi\right\vert \right)  $ in $\mathcal{H}_{in,2}^{\prime}\otimes
\mathcal{H}_{R}^{n\prime}$. Repeating this for $n$ times, composition of
$\tilde{\Gamma}^{n}$ is done.
\end{proof}

\begin{proposition}
Consider the family $\left\{  \Lambda_{\theta}\right\}  _{\theta\in\left\{
+,-\right\}  }$ in Examples\thinspace\ref{ex:ortho-measure} and
\ref{ex:mes-rotate}. Then, a universally optimal input state for parallel
repetition is $\left\vert \Phi_{d}\right\rangle ^{\otimes n}$ (identical repetition).
\end{proposition}

\begin{proof}
If $\ \Lambda_{\theta}$ is in the form of (\ref{meas-ortho}), or of
(\ref{meas-ortho}), so is $\Lambda_{\theta}^{\otimes n}$. Therefore,
$\left\vert \Phi_{d^{n}}\right\rangle =\left\vert \Phi_{d}\right\rangle
^{\otimes n}$ is optimal.
\end{proof}

This proposition motivates following definition of \textit{classical
adaptation}: Given $\Lambda_{\theta}^{\otimes n}$, we divide this into
$\Lambda_{\theta}^{\otimes n_{1}}$, $\Lambda_{\theta}^{\otimes n_{2}}$%
,$\cdots$,$\Lambda_{\theta}^{\otimes n_{m}}$, with $\sum_{i=1}^{m}n_{i}=n$. We
know that preparing input state separately in each block $\left\vert \psi
_{1}\right\rangle \in\mathcal{H}_{in}^{\otimes n_{1}}$, $\left\vert \psi
_{2}\right\rangle \in\mathcal{H}_{in}^{\otimes n_{1}}$,$\cdots$,$\left\vert
\psi_{m}\right\rangle \in\mathcal{H}_{in}^{\otimes n_{m}}$ can achieve the
same as the optimal parallel repetition. So the question arises whether we can
do better by choosing $\left\vert \psi_{j}^{x^{j-1}}\right\rangle $ depending
on the data $x^{j-1}=\left(  x_{1},x_{2},\cdots,x_{j-1}\right)  $ from
measurements $M_{1}$, $M_{2}^{x^{1}}$, $\cdots$,$M_{j-1}^{x^{j-2}}$ applied to
$\left\vert \psi_{1}\right\rangle $, $\left\vert \psi_{2}^{x^{1}}\right\rangle
$, $\cdots$, $\left\vert \psi_{j-1}^{x^{j-2}}\right\rangle $, respectively.
(Note here the measurement at $j$th step is depends on the previous data
sequence $x^{j-1}=\left(  x_{1},x_{2},\cdots,x_{j-1}\right)  $.)

\begin{theorem}
Consider a channel family in Example\thinspace\ref{ex:ortho-measure} or
\ref{ex:mes-rotate}. Then, classical adaptation does not improve identical repetition.
\end{theorem}

\begin{proof}
Let $\vec{M}^{j}:=\left\{  M_{1},M_{2}^{x^{j}},\cdots,M_{j}^{x^{j-1}}\right\}
_{x^{j-1}}$ and $\vec{\psi}^{j}:=\left\{  \left\vert \psi_{1}\right\rangle
,\left\vert \psi_{2}^{x^{1}}\right\rangle ,\cdots,\left\vert \psi
_{j-1}^{x^{j-1}}\right\rangle \right\}  _{x^{j-1}}$. Also, $p_{\theta,\vec
{M}^{m},\vec{\psi}^{m}}\left(  t\right)  $ is the probability of choosing the
decision $t$ when sequence of adaptive measurements $\vec{M}^{m}$ and inputs
$\vec{\psi}^{m}$ are chosen. Then, with the prior distribution $\pi\left(
\theta\right)  $, the minimized risk is
\begin{align*}
&  \inf_{\vec{M}^{m},\vec{\psi}^{m}}\sum_{\theta,t}\pi\left(  \theta\right)
\,p_{\theta,\vec{M}^{m},\vec{\psi}^{m}}\left(  t\right)  l_{\theta}\left(
t\,\right) \\
&  =\inf_{\vec{M}^{m-1},\vec{\psi}^{m-1}}\sum_{x^{m-1}}\,\inf_{M_{m}^{x^{m-1}%
},\rho_{x^{m-1}}}\sum_{t}\sum_{\theta}\pi\left(  \theta\right)  \,\times\\
&  p_{\theta,\vec{M}^{m-1},\vec{\psi}^{m-1}}\left(  x^{m-1}\right)
\mathrm{tr}\,\left\{  \Lambda^{\otimes n_{m}}\otimes\mathbf{I}\left(
\rho_{x^{m-1}}\right)  M_{m}^{x^{m-1}}\left(  t\right)  \right\}  l_{\theta
}\left(  t\,\right)  ,
\end{align*}
Let us denote the marginal distribution of $x^{m-1}$ and conditional
distribution of $\theta$ given $x^{m-1}$ by%
\begin{align*}
\,p_{\pi,\vec{M}^{m-1},\vec{\psi}^{m-1}}\left(  x^{m-1}\right)   &
:=\sum_{\theta}\pi\left(  \theta\right)  \,p_{\theta,\vec{M}^{m-1},\vec{\psi
}^{m-1}}\left(  x^{m-1}\right)  ,\\
\tilde{\pi}_{\vec{M}^{m-1},\vec{\psi}^{m-1}}\,\left(  \theta|x^{m-1}\right)
&  :=\pi\left(  \theta\right)  \,p_{\theta,\vec{M}^{m-1},\vec{\psi}^{m-1}%
}\left(  x^{m-1}\right)  /\,p_{\pi,\vec{M}^{m-1},\vec{\psi}^{m-1}}\left(
x^{m-1}\right)  ,
\end{align*}
respectively. Then, the minimized risk is
\begin{align*}
&  \inf_{\vec{M}^{m},\vec{\psi}^{m}}\sum_{\theta,t}\pi\left(  \theta\right)
\,p_{\theta,\vec{M}^{m},\vec{\psi}^{m}}\left(  t\right)  l_{\theta}\left(
t\right) \\
&  =\inf_{\vec{M}^{m-1},\vec{\psi}^{m-1}}\sum_{x^{m-1}}\,p_{\pi,\vec{M}%
^{m-1},\vec{\psi}^{m-1}}\left(  x^{m-1}\right)  \times\\
&  \inf_{M_{m}^{x^{m-1}},\rho_{x^{m-1}}}\sum_{\theta,t}\tilde{\pi}_{\vec
{M}^{m-1},\vec{\psi}^{m-1}}\,\left(  \theta|x^{m-1}\right)  \mathrm{tr}%
\,\left\{  \Lambda^{\otimes n_{m}}\otimes\mathbf{I}\left(  \rho_{x^{m-1}%
}\right)  M_{m}^{x^{m-1}}\left(  t\right)  \right\}  l_{\theta}\left(
t\right)  .
\end{align*}
By definition of a universally optimal state, infimum over $\rho_{x^{m-1}}$
can be achieved by $\rho_{x^{m-1}}=\left\vert \Phi_{d}\right\rangle
\left\langle \Phi_{d}\right\vert ^{\otimes n_{m}}$ , which does not depends on
the data sequence $x^{m-1}$. Therefore, we can merge the last two steps into
one; depending on $x^{m-2}$, we send $\rho_{x^{m-2}}\otimes\left\vert \Phi
_{d}\right\rangle \left\langle \Phi_{d}\right\vert ^{\otimes n_{m}}$ into
$\Lambda^{\otimes n_{m-1}+n_{m}}\otimes\mathbf{I}$ and apply $M_{m-1}%
^{x^{m-2}}$ and $M_{m}^{x^{m-1}}$, successively. Repeating this process, we
can get rid of classical adaptation.
\end{proof}

\appendix

\section{Proof of Lemma\thinspace\ref{lem:parallel}}

\label{appendix:pf-lem}

\begin{lemma}
\label{lem:polygon}Suppose $x_{i}\geq0$ ($i=1,\cdots,d$) Then, $Ang\left(
x\right)  \neq\emptyset$ if and only if
\begin{equation}
x_{1}^{\downarrow}\leq\sum_{i=2}^{d}x_{i}^{\downarrow}. \label{x<sumx}%
\end{equation}

\end{lemma}

\begin{proof}
Obviously, we only have to prove 'if'. If $d=3$, the assertion follows from
triangle inequality. Suppose the assertion is true for $d-1$, or for any
$y=\left(  y_{1},y_{2},\cdots,y_{d-1}\right)  $ with
\[
y_{1}\geq y_{2}\geq\cdots\geq y_{d-1}%
\]
and
\[
y_{1}\leq\sum_{i=2}^{d-1}y_{i},
\]
we have $Ang\left(  y\right)  \neq\emptyset$. Suppose $x_{d-1}^{\downarrow
}+x_{d}^{\downarrow}\leq x_{1}^{\downarrow}$ and
\[
x_{1}^{\downarrow}\leq\sum_{i=2}^{d}x_{i}^{\downarrow}=\sum_{i=2}^{d-2}%
x_{i}^{\downarrow}+x_{d-1}^{\downarrow}+x_{d}^{\downarrow}.
\]
hold. Then,
\[
y_{1}:=x_{1},\,y_{2}:=x_{2},\cdots,y_{d-2}:=x_{d-2},\,y_{d-1}=x_{d-1}%
^{\downarrow}+x_{d}^{\downarrow},
\]
$Ang\left(  y\right)  \neq\emptyset$ by the hypothesis of induction.
Therefore, $Ang\left(  x\right)  \neq\emptyset$ holds. On the other hand,
suppose $x_{d-1}^{\downarrow}+x_{d}^{\downarrow}>x_{1}^{\downarrow}$. Observe
$x_{1}^{\downarrow}\geq x_{2}^{\downarrow}\geq\cdots\geq x_{d}^{\downarrow
}\geq0$ yields
\[
\sum_{i=1}^{d-2}x_{i}^{\downarrow}\geq x_{d-1}^{\downarrow}+x_{d}^{\downarrow
}.
\]
Therefore, with
\[
y_{1}:=x_{d-1}^{\downarrow}+x_{d}^{\downarrow},\,y_{2}=x_{1}^{\downarrow
},\,y_{3}=x_{2}^{\downarrow},\cdots,y_{d-1}=x_{d-2}^{\downarrow},
\]
by the hypothesis of induction, $Ang\left(  y\right)  \neq\emptyset$ holds.
Therefore, $Ang\left(  x\right)  \neq\emptyset$ holds. After all, we have the assertion.
\end{proof}

\begin{lemma}
\label{lem:triangle}Suppose $d=3$ and\ \ $x_{i}\geq0$, $x_{1}\geq x_{2}\geq
x_{3}$. Then, if $x_{1}<x_{2}+x_{3}$,
\[
Ang\left(  x\right)  =\left\{  \left(  \omega_{1},\omega_{2}\right)  ,\left(
\overline{\omega_{1}},\overline{\omega_{2}}\right)  \right\}  .
\]
If $x_{1}=x_{2}+x_{3}$,
\[
Ang\left(  x\right)  =\left\{  \left(  -1,1,1\right)  \right\}  .
\]

\end{lemma}

\bigskip

\begin{lemma}
\label{lem:smooth}Suppose $d\geq4$, $x_{1}\geq x_{2}\geq\cdots\geq x_{d}>0$,
and
\begin{equation}
x_{1}<\sum_{i=2}^{d}x_{i}. \label{x1<sum-x}%
\end{equation}
Then, $Ang\left(  x\right)  $ is a ($d-3$)-dimensional smooth manifold.
\end{lemma}

\begin{proof}
Let%
\begin{align*}
z_{k}  &  :=\sum_{i=k}^{d-1}x_{i}\omega_{i}+x_{d},\,\ \left(  k=3,\cdots
,d-1\right)  ,\\
z_{d}  &  :=x_{d},\,\,\,r_{k}:=\left\vert z_{k}\right\vert ,\\
\,\vec{r}  &  :=\left(  r_{3},r_{4},\cdots,r_{d-1}\right)  .
\end{align*}
Suppose $\vec{r}$ is fixed. Then, length of each edge of each triangle
$z_{1}z_{k-1}z_{k}$ ($k=3,\cdots,d$) is decided, and $\vec{\omega}$ can take
only finite possible values. Also, the map from $\vec{r}$ $\ $to $\vec{\omega
}$ is smooth. Therefore, we use $\vec{r}$ as a local coordinate of $Ang\left(
x\right)  $. Let $A\left(  x\right)  $ be the set of all $\vec{r}$s such that
$\vec{\omega}\in Ang\left(  x\right)  $. Below, we show the interior $A\left(
x\right)  ^{\circ}$ of $A\left(  x\right)  $ is non-empty. Then, the assertion
of the lemma immediately follows.

An element of $A\left(  x\right)  $ is constructed as follows. We first fix
$\omega_{d-1}$, $r_{d-1}$, then $\omega_{d-2}$, $r_{d-2}$, ..., $\omega_{k+1}%
$, $r_{k+1}$. We choose $r_{k}$ so that the following (\ref{r<r<r}) and
(\ref{x<r<x}) are satisfied (then, $\omega_{k}$ can take only one of two
possible values.); First, by
\begin{equation}
r_{k}=\left\vert z_{k+1}+x_{k}\omega_{k}\right\vert , \label{rk}%
\end{equation}
existence of $\omega_{k}$ is equivalent to
\begin{equation}
\left\vert r_{k+1}-x_{k}\right\vert \leq r_{k}\leq r_{k+1}+x_{k}.
\label{r<r<r}%
\end{equation}
Also, for $\omega_{k-1}$,$\cdots$,$\omega_{1}$ to exist, by Lemma\thinspace
\ref{lem:polygon}, it is necessary and sufficient that
\begin{equation}
x_{1}-\sum_{i=2}^{k-1}x_{i}\leq r_{k}\leq\sum_{i=1}^{k-1}x_{i}. \label{x<r<x}%
\end{equation}
Therefore, $A\left(  x\right)  $ is the set of $\vec{r}$s with (\ref{r<r<r})
and (\ref{x<r<x}) for each $k=3$,$\cdots$,$d-1$. Therefore, $A\left(
x\right)  ^{\circ}$ is the set of all $\vec{r}$s with
\begin{equation}
\left\vert r_{k+1}-x_{k}\right\vert <r_{k}<r_{k+1}+x_{k}, \label{r<r<r-2}%
\end{equation}
and
\begin{equation}
x_{1}-\sum_{i=2}^{k-1}x_{i}<r_{k}<\sum_{i=1}^{k-1}x_{i} \label{x<r<x-2}%
\end{equation}
for each $k=3$,$\cdots$,$d-1$.

This $A\left(  x\right)  ^{\circ}$ is non-empty due to the following reasons.
By (\ref{x1<sum-x}), we have
\[
x_{1}-\sum_{i=2}^{d-2}x_{i}<r_{d}+x_{d-1}.
\]
Also, by $x_{d-2}\geq x_{d-1}\geq r_{d}>0$, we have
\[
\left\vert r_{d}-x_{d-1}\right\vert <\max\left\{  r_{d},x_{d-1}\right\}  \leq
x_{d-2}<\sum_{i=1}^{d-2}x_{i}.
\]
Therefore, combining these, the overlap of the set
\[
\left\{  r_{d-1}\,;\,x_{1}-\sum_{i=2}^{d-2}x_{i}<r_{d-1}<\sum_{i=1}^{d-2}%
x_{i}\right\}  ,
\]
and the set
\[
\left\{  r_{d-1}\,;\,\left\vert r_{d}-x_{d-1}\right\vert <r_{d-1}%
<r_{d}+x_{d-1}\right\}
\]
is not empty. Recursively, suppose $r_{k}$ with (\ref{r<r<r-2}) and
(\ref{x<r<x-2}) exists. Then, by (\ref{r<r<r-2}),%
\[
r_{k}-x_{k-1}<r_{k}<\sum_{i=1}^{k-2}x_{i},\,\,x_{1}-\sum_{i=2}^{k-2}%
x_{i}<r_{k}+x_{k-1}.
\]
Also, by $x_{1}\geq x_{2}\geq\cdots\geq x_{d}$ and $x_{i}>0$,
\[
x_{k-1}-r_{k}\leq x_{k-1}<\sum_{i=1}^{k-2}x_{i}\,.
\]
Therefore,
\[
\left\vert r_{k}-x_{k-1}\right\vert <\sum_{i=1}^{k-2}x_{i},\,x_{1}-\sum
_{i=2}^{k-2}x_{i}<r_{k}+x_{k-1}.
\]
Therefore, there is $r_{k-1}$ with
\[
\left\vert r_{k}-x_{k-1}\right\vert <r_{k-1}<r_{k}+x_{k-1}%
\]
and
\[
x_{1}-\sum_{i=2}^{k-2}x_{i}<r_{k}<\sum_{i=1}^{k-2}x_{i}.
\]
Therefore, there exists $\vec{r}$ such that (\ref{r<r<r-2}) and (\ref{x<r<x-2}%
) hold for each $k$, or equivalently, $A\left(  x\right)  ^{\circ}$ is non-empty.
\end{proof}

\begin{lemma}
\label{lem:other-than-triangle}Suppose $d\geq4$ and $x_{i}>0$ ($i=1,\cdots
,d$). Then $Ang\left(  x\right)  $ contains an element $\vec{\omega}$ such
that the set $\left\{  \omega_{1},\omega_{2},\cdots,\omega_{d-1}\right\}  $
contains at least three distinct elements.
\end{lemma}

\begin{proof}
Suppose $\omega_{1}$, $\omega_{2}$, $\cdots$, $\omega_{d-1}$ can take at most
two distinct values for any element of $Ang\left(  x\right)  $. Let $I$ be a
subset of $\left\{  1,\cdots,d-1\right\}  $ , and $\omega_{i}=\nu_{I}$
($\,i\in I$), $\omega_{i}=\nu_{I}^{\prime}$ ($\,i\in I^{c}$). Then, $\left(
\nu_{I},\nu_{I}^{\prime}\right)  $ is decided by Lemma\thinspace
\ref{lem:triangle}. Moving $I$ over all the subsets of $\left\{
1,\cdots,d-1\right\}  $, $\left(  \nu_{I},\nu_{I}^{\prime}\right)  $ can move
over discretely many values. This contradicts with Lemma\thinspace
\ref{lem:smooth}.
\end{proof}

\bigskip Now, we are in the position to state the proof of Lemma\thinspace
\ref{lem:parallel}.

\begin{proof}
\textbf{(Lemma\thinspace\ref{lem:parallel})} When $d=3$, the assertion is
trivial. So suppose $d\geq4$. Let
\[
\vec{\omega}\left(  t\right)  =\left(  e^{\sqrt{-1}\eta_{1}\left(  t\right)
},e^{\sqrt{-1}\eta_{2}\left(  t\right)  },\cdots,e^{\sqrt{-1}\eta_{d-1}\left(
t\right)  }\right)  \in Ang\left(  x\right)  \subset Ang\left(  x^{\prime
}\right)  ,
\]
where $\eta_{i}\left(  t\right)  $ are smooth functions. (Such smooth
parameter $t$ exists due to Lemma\thinspace\ref{lem:smooth}.) Then,
\[
\sum_{i=1}^{d-1}x_{i}e^{\sqrt{-1}\eta_{i}}+x_{d}=\sum_{i=1}^{d-1}x_{i}%
^{\prime}e^{\sqrt{-1}\eta_{i}}+x_{d}^{\prime}=0.
\]
Differentiating by $t$,
\begin{equation}
\sum_{i=1}^{d-1}x_{i}\overset{\cdot}{\eta}_{i}e^{\sqrt{-1}\eta_{i}}=\sum
_{i=1}^{d-1}x_{i}^{\prime}\overset{\cdot}{\eta}_{i}e^{\sqrt{-1}\eta_{i}}=0.
\label{doteta}%
\end{equation}

Due to Lemma\thinspace\ref{lem:smooth}, with
\[
\tilde{\eta}:=\left(  \overset{\cdot}{\eta}_{1},\overset{\cdot}{\eta}%
_{2},\cdots,\overset{\cdot}{\eta}_{d-1}\right)  ,
\]
$\mathrm{span}\left\{  \tilde{\eta};\text{(\ref{doteta}) holds}\right\}  $ is
$d-3$ dimensional. Therefore, its orthogonal complement in $%
\mathbb{R}
^{d-1}$ is at most two dimensional. By (\ref{doteta}),
\begin{align*}
&  \,(x_{1}\cos\eta_{1},\cdots,x_{d-1}\cos\eta_{d-1}\,),\,(x_{1}\sin\eta
_{1}\,,\cdots,x_{d-1}\sin\eta_{d-1}\,),\\
&  (x_{1}^{\prime}\cos\eta_{1}\,,\cdots,x_{d-1}^{\prime}\cos\eta
_{d-1}\,),\,(x_{1}^{\prime}\sin\eta_{1}\,,\cdots,\,x_{d-1}^{\prime}\sin
\eta_{d-1}),
\end{align*}
are orthogonal to $\mathrm{span}\left\{  \tilde{\eta};\text{(\ref{doteta})
holds }\right\}  $. By Lemma\thinspace\ref{lem:other-than-triangle}, we can
choose $\eta_{i}$ so that the set $\left\{  \eta_{1},\eta_{2},\cdots
,\eta_{d-1}\right\}  $ contains at least three distinct values. Therefore,
$\,(x_{1}\cos\eta_{1},\cdots,x_{d-1}\cos\eta_{d-1}\,)$ and $(x_{1}\sin\eta
_{1}\,,\cdots,x_{d-1}\sin\eta_{d-1}\,)$ are linearly independent, thus can be
chosen as a basis of orthogonal complement of $\mathrm{span}\left\{
\tilde{\eta};\text{(\ref{doteta}) holds}\right\}  $. Therefore, there are
$a_{1}$,$\cdots$,$a_{4}$ with
\begin{align*}
(x_{1}^{\prime}\cos\eta_{1}\,,\cdots,x_{d-1}^{\prime}\cos\eta_{d-1}\,)  &
=a_{1}\,(x_{1}\cos\eta_{1},\cdots,x_{d-1}\cos\eta_{d-1}\,)+a_{2}(x_{1}\sin
\eta_{1}\,,\cdots,x_{d-1}\sin\eta_{d-1}\,),\\
(x_{1}^{\prime}\sin\eta_{1}\,,\cdots,\,x_{d-1}^{\prime}\sin\eta_{d-1})  &
=a_{3}(x_{1}\cos\eta_{1},\cdots,x_{d-1}\cos\eta_{d-1})+a_{4}(x_{1}\sin\eta
_{1}\,,\cdots,x_{d-1}\sin\eta_{d-1}\,).
\end{align*}
Therefore,
\begin{align*}
&  \left(  a_{1}\cos\eta_{i}\sin\eta_{i}+a_{2}\sin^{2}\eta_{i}-a_{3}\cos
^{2}\eta_{i}-a_{4}\cos\eta_{i}\sin\eta_{i}\right)  x_{i}\\
&  =\left(  \frac{a_{1}-a_{4}}{2}\sin2\eta_{i}-\frac{a_{2}+a_{3}}{2}\cos
2\eta_{i}+\frac{a_{2}-a_{3}}{2}\right)  x_{i}=0.
\end{align*}
Since $x_{i}>0$, we have
\[
\left(  a_{1}-a_{4}\right)  \sin2\eta_{i}-\left(  a_{2}+a_{3}\right)
\cos2\eta_{i}+a_{2}-a_{3}=0.
\]
Therefore, $a_{1}-a_{4}=a_{2}+a_{3}=0$, since the set $\left\{  \eta_{1}%
,\eta_{2},\cdots,\eta_{d-1}\right\}  $ contains at least three distinct values
by Lemma\thinspace\ref{lem:other-than-triangle}. Therefore,
\[
a_{1}-a_{4}=a_{2}+a_{3}=a_{2}-a_{3}=0,
\]
which means%
\begin{align*}
(\cos\eta_{1}\,x_{1}^{\prime},\cdots,\cos\eta_{d-1}\,x_{d-1}^{\prime})  &
=a_{1}(\cos\eta_{1}\,x_{1},\cdots,\cos\eta_{d-1}\,x_{d-1}),\\
(\sin\eta_{1}\,x_{1}^{\prime},\cdots,\sin\eta_{d-1}\,x_{d-1}^{\prime})  &
=a_{1}(\sin\eta_{1}\,x_{1},\cdots,\sin\eta_{d-1}\,x_{d-1}).
\end{align*}
Since one of $\cos\eta_{i}$ and $\sin\eta_{i}$ is always non-zero, we have the assertion.
\end{proof}

\bigskip

\section{Proof of Lemma\thinspace\ref{lem:a/b}}

\label{appendix:a/b}

\begin{proof}
\textbf{(Lemma\thinspace\ref{lem:a/b})}

Observe%
\[
\frac{\alpha_{d-1}}{\beta_{d-1}}-\frac{\sum_{i=d-1}^{d}\left\vert \gamma
_{i}\right\vert ^{2}\alpha_{i}}{\sum_{i=d-1}^{d}\left\vert c\gamma\right\vert
^{2}\beta_{i}}=\frac{\left\vert \gamma_{d}\right\vert ^{2}\beta_{d}}%
{\sum_{i=d-1}^{d}\left\vert \gamma_{i}\right\vert ^{2}\beta_{i}}\left(
\frac{\alpha_{d-1}}{\beta_{d-1}}-\frac{\alpha_{d}}{\beta_{d}}\right)  .
\]
Therefore, if $\sum_{i=d-1}^{d}\left\vert \gamma_{i}\right\vert ^{2}\neq0$, we
have
\[
\frac{\alpha_{d-1}}{\beta_{d-1}}\geq\frac{\sum_{i=d-1}^{d}\left\vert
\gamma_{i}\right\vert ^{2}\alpha_{i}}{\sum_{i=d-1}^{d}\left\vert \gamma
_{i}\right\vert ^{2}\beta_{i}}.
\]
Next, observe
\begin{align*}
\frac{\alpha_{d-2}}{\beta_{d-2}}-\frac{\sum_{i=d-2}^{d}\left\vert \gamma
_{i}\right\vert ^{2}\alpha_{i}}{\sum_{i=d-2}^{d}\left\vert \gamma
_{i}\right\vert ^{2}\beta_{i}}  &  =\frac{\sum_{i=d-1}^{d}\left\vert
\gamma_{i}\right\vert ^{2}\beta_{i}}{\sum_{i=d-2}^{d}\left\vert \gamma
_{i}\right\vert ^{2}\beta_{i}}\left(  \frac{\alpha_{d-2}}{\beta_{d-2}}%
-\frac{\sum_{i=d-1}^{d}\left\vert \gamma_{i}\right\vert ^{2}\alpha_{i}}%
{\sum_{i=d-1}^{d}\left\vert \gamma_{i}\right\vert ^{2}\beta_{i}}\right) \\
&  \geq\frac{\sum_{i=d-1}^{d}\left\vert \gamma_{i}\right\vert ^{2}\beta_{i}%
}{\sum_{i=d-2}^{d}\left\vert \gamma_{i}\right\vert ^{2}\beta_{i}}\left(
\frac{\alpha_{d-2}}{\beta_{d-2}}-\frac{\alpha_{d-1}}{\beta_{d-1}}\right)
\end{align*}
Therefore, if $\sum_{i=d-2}^{d}\left\vert \gamma_{i}\right\vert ^{2}\neq0$, we
have%
\[
\frac{\alpha_{d-2}}{\beta_{d-2}}\geq\frac{\sum_{i=d-2}^{d}\left\vert
\gamma_{i}\right\vert ^{2}\alpha_{i}}{\sum_{i=d-2}^{d}\left\vert \gamma
_{i}\right\vert ^{2}\beta_{i}}.
\]
Recursively, if $\sum_{i=2}^{d}\left\vert \gamma_{i}\right\vert ^{2}\neq0$, we
have
\[
\frac{\alpha_{2}}{\beta_{2}}\geq\frac{\sum_{i=2}^{d}\left\vert \gamma
_{i}\right\vert ^{2}\alpha_{i}}{\sum_{i=2}^{d}\left\vert \gamma_{i}\right\vert
^{2}\beta_{i}}.
\]
Observe
\begin{align*}
&  \frac{\alpha_{1}}{\beta_{1}}-\frac{\sum_{i=1}^{d}\left\vert \gamma
_{i}\right\vert ^{2}\alpha_{i}}{\sum_{i=1}^{d}\left\vert \gamma_{i}\right\vert
^{2}\beta_{i}}\\
&  =\frac{\sum_{i=2}^{d}\left\vert \gamma_{i}\right\vert ^{2}\beta_{i}}%
{\sum_{i=1}^{d}\left\vert \gamma_{i}\right\vert ^{2}\beta_{i}}\left(
\frac{\alpha_{1}}{\beta_{1}}-\frac{\sum_{i=2}^{d}\left\vert \gamma
_{i}\right\vert ^{2}\alpha_{i}}{\sum_{i=2}^{d}\left\vert \gamma_{i}\right\vert
^{2}\beta_{i}}\right) \\
&  \geq\frac{\sum_{i=2}^{d}\left\vert \gamma_{i}\right\vert ^{2}\beta_{i}%
}{\sum_{i=1}^{d}\left\vert \gamma_{i}\right\vert ^{2}\beta_{i}}\left(
\frac{\alpha_{1}}{\beta_{1}}-\frac{\alpha_{2}}{\beta_{2}}\right)
\end{align*}
Therefore, due to (\ref{a/b}),
\[
\frac{\alpha_{1}}{\beta_{1}}=\frac{\sum_{i=1}^{d}\left\vert \gamma
_{i}\right\vert ^{2}\alpha_{i}}{\sum_{i=1}^{d}\left\vert \gamma_{i}\right\vert
^{2}\beta_{i}}%
\]
implies $\sum_{i=2}^{d}\left\vert \gamma_{i}\right\vert ^{2}=0$. Thus we have
the assertion.
\end{proof}

\end{document}